\documentclass[a4paper,12pt]{article}
\usepackage[utf8]{inputenc}
\usepackage[T1]{fontenc}
\usepackage[a4paper,includefoot,margin=2.54cm]{geometry}
\usepackage{amssymb}
\usepackage{amsmath}
\usepackage{amsthm}
\newsavebox\mcAcontent
\savebox\mcAcontent{$\mathcal{A}$}
\newcommand\mcA{\usebox{\mcAcontent}}
\usepackage{newtxtext}\usepackage[smallerops]{newtxmath}
\usepackage{mathrsfs}
\usepackage{url}

\usepackage{graphicx}   
\usepackage{mathtools}
\usepackage[dvipsnames]{xcolor}

\renewcommand{\Re}{\mathrm{Re}\,}
\renewcommand{\Im}{\mathrm{Im}\,}
\renewcommand{\epsilon}{\varepsilon}
\renewcommand{\phi}{\varphi}
\renewcommand{\theta}{\vartheta}
\renewcommand{\rho}{\varrho}
\newcommand{\rr}{\mathbb R}  
\newcommand{\rp}{\mathbb R _+}
\renewcommand{\P}{P}
\newcommand{\E}{E}
\newcommand{\eu}{\mathrm{eu}}
\newcommand{\am}{\mathrm{am}}
\newcommand{\comp}{C}
\newcommand{\F}{\mathscr F}
\newcommand{\B}{B}\newcommand{\M}{M}

\newcommand{\cl}{\mathrm{cl}}
\newcommand{\nn}{\mathbb N}  
\newcommand{\cc}{\mathbb C}
\newcommand{\nno}{\phi}
\newcommand{\nnormal}{N}
\newcommand{\T}{\mathbf{T}}
\newcommand{\U}{\mathbf{U}}
\newcommand{\I}{\mathscr{I}}
\newcommand{\lsc}{\mathrm{lsc}}

\newcommand{\inner}{\mathop\textnormal{int}}
\newcommand{\ctheta}{\breve\theta}
\newcommand{\cx}{\breve x}

\newtheorem{theorem}{Theorem}[section]
\newtheorem{proposition}[theorem]{Proposition}
\newtheorem{lemma}[theorem]{Lemma}
\newtheorem{@definition}[theorem]{Definition}
\newenvironment{definition}{\begin{@definition}\rm}{\end{@definition}}
\newtheorem{@example}[theorem]{Example}
\newenvironment{example}{\begin{@example}\rm}{\end{@example}}

\numberwithin{equation}{section}

\newcommand{\represents}[1]{\stackrel{#1}{\longrightarrow}}

\begin{document}

\title{Are American options European after all?}
\author{S\"oren Christensen\footnote{Mathematisches Seminar, Christian-Albrechts-Universit\"at zu Kiel, Kiel, Germany, email: christensen@math.uni-kiel.de}\and
Jan Kallsen\footnote{Mathematisches Seminar, Christian-Albrechts-Universit\"at zu Kiel, Kiel, Germany, email: kallsen@math.uni-kiel.de}\and
Matthias Lenga\footnote{Philips Research Europe, Hamburg, Germany, email: MatthiasLenga@posteo.de}}
\date{}

\maketitle

\begin{abstract}
We call a given American option \emph{representable} if there exists a European claim which dominates the American payoff 
at any time and such that the values of the two options coincide in the continuation region of the American option.
This concept has interesting implications from a probabilistic, analytic, financial, and numeric point of view.
Relying on methods from \cite{jourdain.martini.01,jourdain.martini.02,christensen.14} and convex duality, we
make a first step towards verifying representability of American options.


\end{abstract}

\section{Introduction}\label{s:SEC_intro}
This paper is concerned with reducing the valuation of American options to the simpler problem
of computing prices of European options whose payoff is not path dependent.
For ease of exposition we consider the standard risk-neutral Black-Scholes setting of a deterministic bond and a stock
whose price processes $B$ resp.\ $S=e^X$ evolve according to
\begin{align}\label{e:EQ_BS_market_from_Intro}
\begin{aligned}
dB_t &= r B_t dt, \quad  B_0 = 1,\\
dX_t &= \left( r - \frac{\sigma^2}{2} \right) dt + \sigma dW_t,
\end{aligned}
\end{align}
with parameters $r\geq 0$, $\sigma>0$ and a Wiener process $W$.
Relative to the probability measure $\P_x$, the return process $X$
is assumed to start in $X_0=x$ almost surely.
We denote the fair value of a European option with payoff $f(X_T)$
for a payoff function $f:\rr\to\rr_+$,
time to maturity $T\in\rp$ and initial logarithmic stock price $x$ as 
$v_{\eu,f}(T,x)$, i.e.\ 
\begin{equation}\label{e:veu}
v_{\eu,f}(T,x):= \E_x\bigl( e^{-rT} f(X_T)\bigr).
\end{equation}
Similarly, for an upper semi-continuous payoff function $g:\rr\to\rr_+$ satisfying the integrability condition
\begin{equation}\label{eq:integrability_condition}
\E_x\biggl( \sup_{t \in [0,T]} g(X_t)\biggr)< \infty,
\end{equation}
the fair value of an American claim with payoff process $Z = g(X)$,
time to maturity $T\in\rp$,
and initial stock price $x$ is written as $v_{\am,g}(T,x)$, i.e.\ 
\begin{equation}\label{e:vam}
v_{\am,g}(T,x):= \sup_{\tau \in \mathscr{T}_{[0,T]} }\E_x(e^{-r\tau} g(X_\tau)),
\end{equation}
where $\mathscr{T}_{[0,T]}$ 
denotes the set of $[0,T]$-valued stopping times. 
We write
\begin{equation}\label{e:cont}
C_T:= \bigl\{(\theta,x) \in [0,T]\times \rr: v_{\am,g}(\theta,x) > g(x) \bigr\}
\end{equation}
and 
\[C_T^\comp:=([0,T]\times \rr)\setminus C_T=\bigl\{(\theta,x) \in \rr_+\times \rr: v_{\am,g}(\theta,x) = g(x) \bigr\}\] 
for the continuation region and the stopping region of the American claim, respectively. 

Fix a time horizon $T$ and an initial log price $X_0=x_0$ such that $(T,x_0)$ is contained in $C_T$. 
For this introductory section let us assume that $C_T$ is a connected set.
We say that a European payoff function $f:\rr\to\rr_+$ \emph{represents} the American payoff function $g:\rr\to\rr_+$ 
if the value of $f$ dominates the value of $g$ everywhere and the two coincide in the continuation
region of the American claim, i.e.\
$v_{\eu,f}(\theta,x) \geq v_{\am,g}(\theta,x)$ 
for all $(\theta,x) \in [0,T] \times \rr$
and $v_{\am,g}(\theta,x) = v_{\eu,f}(\theta,x)$ holds for all $(\theta,x) \in C$.

The main question in this paper is the following:
given an American payoff function $g$, is there a European payoff function $f$ representing $g$?
In this case we call $g$ \emph{representable}.
If representability holds, this has several interesting consequences.
\begin{itemize}
 \item The American value function can be computed efficiently by means of linear programming, as is explained below.
 \item The early exercise boundary can be obtained numerically at low computational costs, see \cite[Section 3.4]{lenga.17}.
 \item  A buy-and-hold position in the European option with time-$T$ payoff $f(X_T)$ hedges the American claim perfectly.
 Put differently, the American option can be hedged statically with a portfolio of calls/puts 
 that does not cost more than the American claim itself. Here, \emph{portfolio} is to be understood in the limiting sense
 of e.g.\ \cite{nachman.88}.
 \item In the continuation region, the difference 
 $v_{\am,g}-v_{\eu,g}$ is the fair value of a European payoff with time-$T$ payoff $f(X_T)-g(X_T)$.
 Put differently, the early exercise premium of the American option can be interpreted as the price of a 
 European claim with a specific payoff profile.
 \item The Snell envelope corresponding to the American option allows for a Markovian-style decomposition, cf.\
 (\ref{e:pseudodoob}) below.
 \item Some analytical properties of the early exercise curve can be obtained easily.
 Indeed, it coincides with the boundary of the set 
 $\{(\theta,x) \in (0,T]\times \rr: v_{\eu,f}(\theta,x) = g(x) \}$. 
 This allows to derive smoothness of the early exercise curve from 
 the analyticity of $v_{\eu,f}$ and the implicit function theorem.
In the same vein, certain analyticity properties of the European payoff function $v_{\eu,f}$ transfer to the American payoff function $v_{\am,g}$.
 \item The solution of the free boundary problem associated to the American option can be extended
 to a solution of the Black-Scholes partial differential equation beyond the free boundary.
\end{itemize}

On top of representability of a given option one may ask 
how to obtain the representing European payoff, at least numerically.
Moreover, are possibly all American options representable?
Or, if this is not the case, do  representable options exist at all -- except for the obvious case where
 early exercise is suboptimal and hence
$g$ itself represents $g$?

The concept of representability is not studied here for the first time. 
It was considered in two seminal papers by Jourdain and Martini, which have not yet received the 
attention they deserve.
In \cite{jourdain.martini.01} it is shown that many European payoffs represent some American payoff, which is obtained in a natural way.
Indeed, given some European payoff function $f$, they define
an American payoff function $\am_{T}(f):\rr\to\rr_+$ as 
\begin{equation}\label{DEF_embedded_am_opt}
\am_{T}(f)(x) := \inf_{\theta \in [0,T]} v_{\eu, f}(\theta,x),
\end{equation}
from now on called the \emph{embedded American option (EAO)} associated with $f$.
If the infimum in (\ref{DEF_embedded_am_opt}) is attained in a connected curve,
$f$ represents its embedded American option $\am_{T}(f)$, cf.\ \cite[Theorem 5]{jourdain.martini.01}.
Jourdain and Martini provide an explicit example where this is the case. On the other hand, they show that
embedded American payoff functions satisfy certain analyticity properties, cf.\ \cite[Proposition 16]{jourdain.martini.01}.
From their results we conclude that representable options exist but that not all American payoff functions are representable.

In their follow-up article \cite{jourdain.martini.02} they study the American put option in detail.
They show that it cannot be represented by any of a seemingly general and reasonable candidate family of 
European claims. This suggests that this particular option may not be representable.
Summing up, Jourdain and Martini provide a way to obtain an American payoff function $g$ that is
represented by a given European claim $f$. 
Our question here is rather the converse: given $g$, is there a representing European claim $f$, and how can it be
obtained?

In order to tackle these problems, we make use of the approach in \cite{christensen.14}.
Fix an American payoff function
$g:\rr\to\rr_+$.
The key contribution of \cite{christensen.14} is the linear optimisation problem
\begin{equation}\label{e:CDEO_general_opt_problem}
\begin{aligned}
& \text{minimise}   & & v_{\eu,f}(T,x_0) \\
& \text{subject to} & & f:\rr\to\rr_+ \text{ measurable and}\\
& & & v_{\eu,f}(\theta,x) \geq g(x) \text{ for all } (\theta,x) \in [0,T] \times \rr.
\end{aligned}
\end{equation}
We call the minimiser $f$ of (\ref{e:CDEO_general_opt_problem})
\emph{cheapest dominating European option (CDEO)} of $g$ relative to $(T,x_0)$.
The linear problem (\ref{e:CDEO_general_opt_problem}) can be solved efficiently by numerical methods, cf.\ \cite{christensen.14} for details.
It is easy to see that the fair price of a CDEO $f$ provides an upper bound to the value of the given American claim $g$. 

However, in \cite{christensen.14} it remains open how large the gap between the two actually is.
While there is a priori no reason why the two should coincide, numerical studies in \cite{christensen.14} 
indicate that the difference seems to be small.
In the present paper, we use the CDEO 
as a candidate which may generate the desired American payoff $g$.
Indeed, if $g$ is representable at all, it must be represented by its CDEO.
This also answers the question how to obtain a representing European payoff function numerically if it exists at all.

It is important to distinguish the minimisation problem (\ref{e:CDEO_general_opt_problem}) and 
more generally the present study from
the well-known duality approaches put forward by 
\cite{rogers.02,davis.karatzas.94,haugh.kogan.04}. 
Consider again an American payoff function $g:\rr \to\rr_+$ leading to the discounted exercise process $\widehat Z_t:=e^{-rt}g(X_t)$.
From \cite{rogers.02} we know that
\begin{equation}\label{e:rogers}
 v_{\am, g}(T,x_0) = \inf\left\{\E_{x_0}\biggl( \sup_{t\in[0,T]}(  \widehat{Z}_t - M_t) \biggr): M \text{ martingale with }M_0=0\right\}.
\end{equation}
Indeed, the inequality $\leq$ is obvious because
\[\E_{x_0}( \widehat{Z}_\tau)=\E_{x_0}(  \widehat{Z}_\tau - M_\tau)
\leq \E_{x_0}\biggl( \sup_{t\in[0,T]}(  \widehat{Z}_t - M_t) \biggr)\]
for any $[0,T]$-valued stopping time $\tau$ and any martingale $M$ with $M_0=0$.
For the converse inequality consider 
the Doob-Meyer decomposition 
\begin{equation}\label{e:doob}
V=V_0+M^V-A^V
\end{equation}
of the Snell envelope $V$ of the discounted exercise process $\widehat Z$,
i.e.\ $M^V$ is a martingale and $A^V$ an increasing process with $M^V_0=0=A^V_0$.
Since 
\[\widehat Z_t-M^V_t\leq V_t-M^V_t= V_0-A^V_t\leq V_0=v_{\am, g}(T,x_0)\]
for any $t\in[0,T]$,
we conclude that the inequality $\geq$ holds in (\ref{e:rogers}) as well.

Similarly, observe that
\begin{multline}\label{e:rogers2}
 v_{\am, g}(T,x_0) 
 = \inf\bigl\{\E_{x_0}(Y) : Y\geq0 \text{ random variable with }\widehat Z_t\leq \E_{x_0}(Y|\F_t), t\in[0,T]\bigr\}.
\end{multline}
Again the inequality $\leq$ is obvious because any martingale dominating $\widehat Z$ is an upper bound
of the discounted American option price process. The converse inequality $\geq$ follows from choosing $Y=V_0+M^V_T$,
where $V$ and $M^V$ are defined as above.

The linear problem (\ref{e:CDEO_general_opt_problem}) can be rephrased as
\begin{equation}\label{e:soeren}
\inf\bigl\{\E_{x_0}(e^{-rT}f(X_T))  : f\!:\!\rr\!\to\!\rr_+ \text{ with }\widehat Z_t\leq \E_{x_0}(e^{-rT}f(X_T)|\F_t), t\in[0,T]\bigr\},
\end{equation}
which seems almost identical to the right-hand side of (\ref{e:rogers2}).
However, the dominating European payoff $Y$ in (\ref{e:rogers2}) may well be path dependent,
which is not the case in (\ref{e:soeren}).
And indeed, it is easy to see that the terminal value $V_0+M^V_T$ cannot typically written as a function
of $X_T$, e.g.\ in the case of an American put. Therefore, the identities
(\ref{e:rogers}) and (\ref{e:rogers2}) do not help in deciding whether the value of the CDEO in the sense of 
(\ref{e:CDEO_general_opt_problem}) coincides
with the price of the given American option $g$.

From a different perspective, one may note that the martingale 
in the Doob-Meyer decomposition (\ref{e:doob}) is not the only one that leads to
optimal choices in (\ref{e:rogers}) and (\ref{e:rogers2}).
In fact, we could replace $M^V$ by $\widetilde M$ in any decomposition of the form
\begin{equation}\label{e:pseudodoob}
 V=V_0+\widetilde M-\widetilde A
\end{equation}
with some martingale $\widetilde M$ and some nonnegative process $\widetilde A$ satisfying $\widetilde M_0=0=\widetilde A_0$.
Contrary to the unique decomposition (\ref{e:doob}) we do not require $\widetilde A$ to be increasing.
As noted above, (\ref{e:soeren}) coincides with the American option price (\ref{e:rogers2})
if we can choose $\widetilde M$ such that $V_0+\widetilde M_T=e^{-rT}f(X_T)$ for some deterministic function $f$.
In this case, the decomposition (\ref{e:pseudodoob}) is of \emph{Markovian style}  in the sense that
both $\widetilde M_t$ and $\widetilde A_t$ are functions of $t$ and $X_t$ at any time $t$.
Hence the issue of representability is linked to the existence of Markovian-style decompositions
(\ref{e:pseudodoob}) of the Snell envelope corresponding to the optimal stopping problem.

The present study serves different purposes.
In Section \ref{s:SEC_Representable_options} 
we establish the link between embedded American options from \cite{jourdain.martini.01}, cheapest dominating 
European options from \cite{christensen.14}, and representability. By providing an example, 
we show that representability may depend on the time horizon $T$, cf.\ Section \ref{su:subsection_am_opt_generated_by_put}.
The main contribution of this paper is contained in Section \ref{s:section_verification_theorem}.
Firstly, we establish the existence of CDEOs in a distributional sense
for sufficiently regular American payoff functions $g$.
Secondly and more importantly, we provide a sufficient criterion for representability of a given American claim.
The assumptions of this result depend on qualitative properties of the corresponding CDEO.
Numerical computations suggest that they are satisfied for the American put, cf.\ Section \ref{s:SEC_Example_am_put}.

Let us fix some notation that is used in the paper. $\Vert \mu  \Vert$ stands for the total variation of a signed measure $\mu$.
The set of signed measures of finite variation on a measurable space $(S,\mathscr S)$ is written as $\M(S)$. 
The vector spaces of real-valued continuous functions and continuous functions vanishing at infinity on $S$ are denoted by $C(S)$ and $C_0(S)$, respectively.
They are Banach spaces with respect to the norm $\Vert\cdot\Vert_\infty$ which generates the topology of uniform convergence $\mathscr{T}_{\mathrm{uc}}$.
By $\M^+(S), C^+(S), C_0^+(S)$ we denote the cones of nonnegative elements in the respective spaces.
The closure and the interior of a set $M$ in some topological space are denoted by $\cl\, M$ and $\inner M$.
We write $\partial M := \cl M \setminus \inner M$ for the boundary of the set.
$B_V(x,r) := \{ v \in V :\Vert v -x \Vert \leq r\}$ denotes the ball with radius $r$ around $x$ in a normed space $V$. If the space is obvious, we simply write
$B(x,r)$.
The Dirac measure in $x$ is denoted as $\delta_x$. 
Moreover we write $\phi(\mu,\sigma^2, \cdot)$ for the probability density function of the normal distribution $N(\mu,\sigma^2)$  with mean $\mu$ and variance $\sigma^2$.
The cumulative distribution function of $N(0,1)$ is denoted  as $\Phi$.
The gradient of a real- or complex-valued function $f$ is denoted as $Df$ and its partial derivatives with respect to its first, second, $d$th argument are written as $D_1f,D_2f,D_df$ etc.
The convex conjugate and the biconjugate in the sense of \cite{rockafellar.74} of a function $v$ are denoted by $v^*$ and $v^{**}$, respectively.

\section{Representable options}\label{s:SEC_Representable_options} 
In this section we derive some general results about embedded, cheapest dominating, and representable options.
For ease of exposition, we focus on the univariate Black-Scholes market \eqref{e:EQ_BS_market_from_Intro}.
Moreover,
we use the notation (\ref{e:veu}, \ref{e:vam}) from Section \ref{s:SEC_intro} for the fair values of European and American options.

\subsection{Embedded American and cheapest dominating European options}\label{s:EAO_CDEO}
Fix $T\in[0,\infty]$.  Let $f:\rr\to\rp$ denote a measurable European payoff function with
$v_{\eu, f}(\theta,x)<\infty$ for all $(\theta,x)\in[0,T]\times \rr$ with $\theta<\infty$
and $g:\rr\to\rp$ an upper semi-continuous American payoff function
satisfying \eqref{eq:integrability_condition}.
Let us recall the following key notions from the introduction.
\begin{definition}\label{d:DEF_EAO_CDEO}
\begin{enumerate}
 \item The \emph{embedded American option (EAO)} of $f$ up to $T$ is defined as the payoff function $\am_{T}(f):\rr\to\rr_+$ given by  
		\begin{equation}\label{e:EAO}
		\am_{T}(f)(x) := \inf\big\{v_{\eu, f}(\theta,x): \theta \in [0,T], \theta<\infty\big\},\quad x\in\rr.
		\end{equation}
 \item  We say that $f$ \emph{superreplicates} or \emph{dominates} $g$ up to $T$ 
 if $v_{\eu,f}(\theta,x) \geq g(x)$ holds for all finite $\theta \in [0, T]$ and $x\in \rr$.
 \item If $T<\infty$ and an initial logarithmic stock price $X_0= x_0$ is given, 
 we call a European payoff function $f^\star$ \emph{cheapest dominating European option (CDEO)} of $g$ relative to $(T,x_0)$
 if $f^\star$ superreplicates $g$ up to $T$ and $v_{\eu,f^\star}(T,x_0) \leq v_{\eu,f}(T,x_0)$ holds for all European payoff 
 functions $f$ dominating $g$ up to time $T$.
 The set of all such CDEOs is denoted as  $\eu_{T,x_0}(g)$.
 We write $\eu_{T,x_0}(g)=f^\star$ if there is a unique CDEO $f^\star$, i.e.\ if $\eu_{T,x_0}(g)=\{f^\star\}$. Here we identify functions which differ only on a set of zero Lebesgue measure.
\end{enumerate}
\end{definition}

We state some first results.
\begin{proposition}\label{p:PROPOSITION_1}
\begin{enumerate}
 \item  The set $\eu_{T,x_0}(g)$ is convex. 
 \item If $f$ superreplicates $g$ up to time $T$, we have
 \begin{equation}\label{e:glessf}
 g(x) \leq v_{\am,g}(\theta,x) \leq v_{\eu,f}(\theta,x)  
 \end{equation}
 for all finite $\theta\in[0,T]$ and all $x\in\rr$
 and in particular
 \begin{equation}\label{e:glessf2} 
  g \leq \am_{T}(f) \leq f.
 \end{equation}
\item 
$g\leq \am_{T}(\eu_{T,x_0}(g))$
in the sense that $g$ is dominated by any element of the right-hand side.  
\item $\am_{T}(f)(x)$ is decreasing in $T$.
\item $\am_{T}(f)(x)$ is  increasing in $f$.
\item If $f$ is upper semi-continuous, so is $x\mapsto\am_{T}(f)(x)$.
\end{enumerate}
\end{proposition}
\begin{proof}
\begin{enumerate}
\item Choose $f_1, f_2 \in \eu_{T,x_0}(g)$ and note that for any $\lambda \in (0,1)$ the convex combination $f_\lambda := \lambda f_1 + (1-\lambda) f_2 $ 
superreplicates $g$ up to $T$. Moreover, we have $v_{\eu,f_\lambda}(T,x_0) = \lambda v_{\eu,f_1}(T ,x_0)  + (1-\lambda) v_{\eu,f_2}(T ,x_0)  = v_{\eu,f_1}(T ,x_0)$,
which implies that the payoff $f_\lambda$ is indeed contained in $\eu_{T,x_0}(g)$.  
\item Recall that the discounted value process 
$V^{(\theta)}=(e^{-rt} v_{\eu,f}(\theta-t, X_t))_{t\in[0,\theta]}$
of the European option with time-$\theta$ payoff $f(X_\theta)$
is a martingale. Indeed, applying the Markov property yields 
\[ e^{-rt} v_{\eu,f}(\theta-t, X_t) = e^{-r\theta} \E_{X_t} \!\left( f(X_{\theta - t}) \right) = e^{-r\theta} \E_{x} \!\left(  f(X_\theta) \middle\vert \F_{t}  \right)\]
for any $t \in [0, \theta]$.

Owing to the superreplication property and the optional sampling theorem, we have
\begin{align*}
v_{\am,g}(\theta,x) 
&= \sup_{\tau \in \mathscr{T}_{[0,\theta]} }\E_x\!\left( e^{-r\tau} g(X_\tau) \right)\\
&\leq \sup_{\tau \in \mathscr{T}_{[0,\theta]} }\E_x \!\left( e^{-r\tau}  v_{\eu,f}(\theta-\tau, X_\tau) \right)\\
&= v_{\eu,f}(\theta,x),
\end{align*}
which proves \eqref{e:glessf}.
Minimising both sides of this inequality with respect to $\theta$ yields \eqref{e:glessf2}.
\item This follows from the fact that any payoff in $\eu_{T,x_0}(g))$ superreplicates $g$ up to time $T$.
\item This is obvious.
\item This is obvious as well.
\item By dominated convergence, 
\[x\mapsto v_{\eu,f}(\theta, x)=e^{-r\theta}\int \phi\bigl(x+(r-\sigma^2/2)\theta,\sigma^2\theta,y\bigr)f(y)dy\] 
is upper semi-continuous in $x$ for 
finite $\theta\leq T$. Since the pointwise infimum of a family of upper semi-continuous functions is upper semi-continuous, the assertion follows.
\qedhere
\end{enumerate}
\end{proof}

Now we turn to the representability of an American claim as explained in Section \ref{s:SEC_intro}.
To this end, we fix $T\in(0,\infty)$ and assume that the continuation region $C_T$ in \eqref{e:cont} is nonempty.
Given any $(T_0,x_0)\in C_T$ we denote by $C_{T_0,x_0}$ the connected component of $C_{T_0} =C_T\cap([0,T_0]\times\rr)$ which contains $(T,x_0)$.

\begin{definition}
 We say that the European payoff function $f$ {\em represents} $g$ relative to $(T_0,x_0)\in C_T$ if
 $f$ superreplicates $g$ up to time $T_0$ and 
 $v_{\am,g}(\theta,x) = v_{\eu,f}(\theta,x)$
 holds for all $(\theta,x)\in C_{T_0,x_0}$.
 In this case we write
 \[ f \represents{T_0,x_0} g.\]
 We call $g$ {\em representable} relative to $(T_0,x_0)$ if there exists some $f$ representing it. 
\end{definition}
If an American payoff is representable, it is in fact represented by its CDEO:
\begin{proposition}\label{PROPOSITION_2}
Suppose that the American payoff function $g$ is continuous and satisfies the growth condition
$g(x) \leq C(1+\vert x \vert^k)$, $x\in\rr$ for some constants $C,k<\infty$. Let $(T_0,x_0)\in C_T$.
If $f \represents{T_0,x_0} g $, the following holds.
\begin{enumerate}
 \item For any $(\widetilde T,\widetilde x)\in C_{T_0,x_0}$
we have
 $f \represents{\widetilde T,\widetilde x} g $.
 \item The representing function is unique up to a Lebesgue-null set, i.e.\ $\widetilde{f} \represents{T_0,x_0} g$ implies 
  $\widetilde{f} = f$ Lebesgue-almost everywhere. 
 \item We have $f = \eu_{T_0,x_0}(g)$ Lebesgue-almost everywhere. 
 \item We have $g(x) = \am_{T_0}(f)(x)$ and hence
 \[g(x) = \am_{T_0}(\eu_{T_0,x_0}(g))(x)\]
 for all $x\in \cl\, \pi(C_{T_0,x_0})$,
  where 
  \[\pi(C_{T_0,x_0}):=\{x\in\rr: (\theta,x)\in C_{T_0,x_0}\mbox{ for some }\theta\in\rr_+\}\] 
  denotes the projection of 
  the set on its second coordinate. 
\item The set $C_{T_0,x_0}$  is a connected component of the continuation region $C'_{T_0}$ associated to the American value function $v_{\am,\am_{T_0}(f)}$.
We have $f \represents{T_0,x_0} \am_{T_0}(f)$ and therefore 
\[f= \eu_{T_0,x_0} (\am_{T_0} (f)).\]
\item Suppose that $\widetilde g\leq g$ is an upper semi-continuous American payoff function with   
  $\widetilde g(x) = g(x)$ for all $x \in \cl\,\pi(C_{T_0,x_0})$.  Then $f \represents{T_0,x_0} \widetilde g$ and
  $C_{T_0,x_0}$ is a connected component of the continuation region $\widetilde C_{T_0}$ associated to the American value function $v_{\am,\widetilde g}$.
\end{enumerate}
\end{proposition}
\begin{proof}
\begin{enumerate}
 \item 
 This is obvious because $C_{(\widetilde{T_0},\widetilde{x})}$ is a subset of $C_{T_0,x_0}$.
\item Assume that $f$ and $\widetilde{f}$ represent $g$ relative to $T_0,x_0$. Clearly, we have $v_{\eu,f}(\theta,x) =  v_{\eu,\widetilde{f}}(\theta,x) = v_{\am, g}(\theta,x)< \infty $ 
 for any $(\theta,x)\in C_{T_0,x_0}$. Lemma \ref{LEMMA_analyticity_sol_heat_eaqtion} implies that the value functions $v_{\eu,f}$ and $v_{\eu,\widetilde{f}}$ 
 have an analytic extension on some $\cc^2$-domain containing the set $(0,T_0) \times \rr$.
The set $C_{T_0,x_0}$ contains an open ball $B$. First, we apply the identity theorem to the variable $\theta$ which shows
that the mappings $v_{\eu,f}$ and $v_{\eu,\widetilde{f}}$ coincide on the open strip $(0,T_0)\times \pi(B)$. 
Then we apply the identity theorem to the variable $x$ which yields $v_{\eu,f}(\theta,x) =  v_{\eu,\widetilde{f}}(\theta,x) < \infty$  
for any $(\theta,x)\in  (0,T_0)\times \rr$. 
Consequently, it is easy to see that the functions
\begin{align*}
u(y) &:=  \nno\bigl(x_0 +\widehat{r} \theta_0 , \sigma^2 \theta_0 , y \bigr) f(y),\\
\widetilde{u}(y) &:=  \nno\bigl(x_0 +\widehat{r} \theta_0 , \sigma^2 \theta_0 , y \bigr) \widetilde{f}(y),
\end{align*}
are both integrable on $\rr$, where we set $\theta_0 :=  T_0/2$ and $\widehat{r}:= r - \sigma^2 /2$.
Lemma \ref{LEMMA_norm_and_quotient_of_normal_pdf}(2) yields 
\begin{align*}
v_{\eu,f}(\theta_0/2, x/2)  
&=   \int  \frac{\nno\!\left( x / 2 +  \widehat{r} \theta_0/2, \sigma^2 \theta_0/2, y \right)}{\nno\!\left(x_0 +\widehat{r} \theta_0 , \sigma^2 \theta_0 , y \right) } u(y) dy \\
&=  \sqrt{2}  \exp\!\left(  \frac{( x_0 - x/2  + \widehat{r} \theta_0/2 )^2}{\sigma^2 \theta_0}  \right)   \int \exp\!\left( -\frac{(y - x + x_0)^2}{2\sigma^2 \theta_0} \right)  u(y) dy
\end{align*}
for any $x \in \rr$. This equation remains valid after replacing $f$ and $u$ by $\widetilde{f}$ and $\widetilde{u}$, respectively.
The mappings $v_{\eu,f}$ and $v_{\eu,\widetilde{f}}$ coincide on $(0,T_0)\times \rr$ and consequently
\[\int \nno\bigl(  x_0, \sigma^2 \theta_0, x-y   \bigr) u(y) dy =  \int\nno\bigl( x_0, \sigma^2 \theta_0, x-y   \bigr)  \widetilde{u}(y) dy \]
holds  for any $x\in\rr$. We multiply both sides of this equation by $e^{i zx}, z\in \rr$ and integrate the $x$ variable over the real line.
After a few simplifications we obtain  
\[  \int e^{i zy} u(y) dy =  \int  e^{i zy}  \widetilde{u}(y) dy \]
for any $z\in\rr$.  The injectivity of the Fourier transform on $L_1(\rr)$ yields that $u,\widetilde{u}$ and therefore
$f,\widetilde{f}$ coincide up to a Lebesgue-null set.
\item  
The growth condition on $g$ implies \eqref{eq:integrability_condition}. 
Moreover, $v_{\am,g}$ is continuous on $[0,T]\times\rr$ by \cite[Theorem 4.1.1]{lamberton.09}.
Observe that $C^C_T$ is closed and
\begin{align}\label{e:optimal_stopping_time}
\tau_T:={}&\inf\{t\geq0:(T-t,X_t)\not\in C_T\}\wedge T\nonumber\\
=&\inf \left\{ t\geq 0 :v_{\am,g}(T-t,X_t) = g(X_t)  \right\}\wedge T
\end{align}
is an optimal stopping time for the stopping problem in (\ref{e:vam}), cf.\ \cite[Corollary 2.9]{peskir.shiryaev.06}.

By \eqref{e:glessf} the function $f$ is contained in $\eu_{T_0,x_0}(g)$.
It remains to be shown that this set is a singleton.
To this end choose a function $h \in \eu_{T_0,x_0}(g)$ and note that $v_{\eu,h}(T_0,x_0) = v_{\eu,f}(T_0,x_0) =  v_{\am, g}(T_0,x_0)< \infty$.
By Lemma \ref{LEMMA_analyticity_sol_heat_eaqtion} the mappings $v_{\eu,h}$ and $v_{\eu,f}$ are analytic on a $\cc^2$-domain containing the set $(0,T_0)\times\rr$.
Define 
\[N:=\bigl\{(\theta,x)\in C_{T_0,x_0}\cap((0,T_0)\times\rr):v_{\eu,h}(\theta,x)\not=v_{\eu,f}(\theta,x) \bigr\},\]
which is open because both $v_{\eu,h}$ and $v_{\eu,f}$ are continuous on $C_{T_0,x_0}$.
Moreover, let $\tau_{T_0}$ be the corresponding optimal stopping time as in \eqref{e:optimal_stopping_time}.

Assume by contradiction that there is an interior point $(\theta_0,\xi_0)$ of $N\subset C_{T_0,x_0}$,
we have 
\begin{equation}\label{e:inN}
 \{\theta_0\}\times[\xi_0-\epsilon,\xi_0+\epsilon]\subset N\subset C_{T_0,x_0}
\end{equation}
and $\P_{x_0 }(\tau_{T_0}>t, X_t \in [\xi_0-\epsilon,\xi_0+\epsilon]) > 0$
for some $\epsilon >0$ and sufficiently small $t<T_0-\theta_0$
because $C_{T_0,x_0}$ is connected and open in $[0,T_0]\times\rr$.
\eqref{e:inN} implies $[\theta_0,T_0]\times[\xi_0-\epsilon,\xi_0+\epsilon]\subset C_{T_0,x_0}$.
From the properties of Brownian motion it also follows that the probability of $X$ staying in the interval
$[\xi_0-\epsilon,\xi_0+\epsilon]$ from time $t$ till $T_0-\theta_0$ is strictly positive.
Hence 
\begin{equation}\label{e:PN}
\P_{x_0 }( (T_0-\tau_{T_0})\vee\theta_0, X_{(T_0-\theta_0)\wedge \tau_{T_0}} ) \in N) >0.       
\end{equation}
On the other hand, we have
\begin{align*}
 \MoveEqLeft{\E_{x_0}\!\left( e^{-r((T_0-\tau_{T_0})\vee\theta_0)}
 \left( v_{\eu,h} -  v_{\eu,f}\right)\bigl((T_0-\tau_{T_0})\vee\theta_0, X_{(T_0-\theta_0)\wedge \tau_{T_0}}\bigr)  \right)}\\
&= \E_{x_0}\!\left( e^{-r((T_0-\tau_{T_0})\vee\theta_0)}
 \left( v_{\eu,h} -  v_{\am,g}\right)\bigl((T_0-\tau_{T_0})\vee\theta_0, X_{(T_0-\theta_0)\wedge \tau_{T_0}}\bigr)  \right)\\
&=  v_{\eu,h}\!\left(T_0, x_0\right) - v_{\am,g}\!\left(T_0, x_0\right)= 0.
\end{align*}
The second equality follows from the fact that the discounted European value process as well as the optimally
stopped Snell envelope of the discounted exercise price 
process are martingales, see \cite[Theorem 2.4 and Remark 2.6]{peskir.shiryaev.06}.
Since it is nonnegative, we conclude that
\[\left( v_{\eu,h} -  v_{\eu,f}\right)\bigl((T_0-\tau_{T_0})\vee\theta_0, X_{(T_0-\theta_0)\wedge \tau_{T_0}}\bigr)=0\quad P_{x_0}\text{-almost surely}\]
in contradiction to \eqref{e:PN}. Hence $N$ is empty.

By the proof of the second assertion we conclude that $h$ and $f$ coincide up to a Lebesgue-null set.
\item Choose any $x \in \pi(C_{T_0,x_0})$ and a $\theta_C \in (0, T_0]$ such that $(\theta_C,x) \in C_{T_0,x_0}$.
Due to compactness, there is a largest $\theta_S \in [0, \theta_C)$ such that $(\theta_S,x)$ is contained in the stopping region.
In view of \cite[Theorem 4.1.1]{lamberton.09},  $v_{\am,g}$ is continuous and therefore
\[ g(x) \leq \am_{T_0}(f)(x) 
\leq \liminf_{\theta \downarrow \theta_S} v_{\eu,f}(\theta,x) 
= \liminf_{\theta \downarrow \theta_S} v_{\am,g}(\theta,x) 
= v_{\am,g}(\theta_S,x) = g(x).\]
This proves the assertion for $x \in \pi(C_{T_0,x_0})$. 
For any $x_b \in \partial \pi(C_{T_0,x_0})$ there is some $\theta_b \in (0,T_0]$ such that $(\theta_b,x_b)$ is  in the boundary of
the set $C_{T_0,x_0}$. For an approximating sequence 
$C_{T_0,x_0} \ni (\theta_n,x_n) \to (\theta_b,x_b)$ as $n \to \infty$ we have
\begin{equation*}
\begin{aligned}
g(x_b) = v_{\am,g}(\theta_b,x_b) =  \liminf_{n\to\infty} v_{\am,g}(\theta_n,x_n)= \liminf_{n\to\infty} v_{\eu,f}(\theta_n,x_n).
\end{aligned}
\end{equation*}
Applying Fatou's lemma we obtain
\[\liminf_{n\to\infty} v_{\eu,f}(\theta_n,x_n) \geq v_{\eu,f}(\theta_b,x_b) \geq \am_{T_0}(f)(x_b) \geq g(x_b),\]
which yields $g(x_b) = \am_{T_0}(f)(x_b)$.
\item The European payoff $f$ superreplicates $\am_{T_0}(f)$ up to time $T_0$. Owing to Proposition \ref{p:PROPOSITION_1}(2), we have $g(x) \leq \am_{T_0}(f)(x)$ and hence
\begin{equation}
\label{EQ_ungl_vam_veu}
v_{\am,g}(\theta,x) \leq v_{\am,\am_{T_0}(f)}(\theta,x)  \leq v_{\eu,f}(x)
\end{equation}
for any $(\theta,x) \in [0,T_0] \times \rr$. Moreover, equality in \eqref{EQ_ungl_vam_veu} holds on the set $C_{T_0,x_0}$ because the payoff
$f$ represents $g$ relative to $(T_0,x_0)$. For any $(\theta,x)\in C_{T_0,x_0}$ the fourth assertion warrants that $g(x) = \am_{T_0}(f)(x)$ and therefore
\[ \am_{T_0}(f)(x) = g(x) < v_{\am,g}(\theta,x) = v_{\am,\am_{T_0}(f)}(\theta,x).\]
This shows that $C_{T_0,x_0}$ is a connected subset of $C_{T_0}'$.
For any boundary point $(\theta,x) \in \partial C_{T_0,x_0}$ with $\theta>0$
we have $g(x) = v_{\am,g}(\theta,x) = v_{\eu,f}(\theta,x)$. In view of \eqref{EQ_ungl_vam_veu}, we obtain 
\[v_{\am,\am_{T_0}(f)}(\theta,x)  \leq v_{\eu,f}(x) = g(x) \leq \am_{T_0}(f)(x), \]
which shows that $(\theta,x)$ is located in the stopping region of the American payoff $\am_{T_0}(f)$.
Summing up, the set $C_{T_0,x_0}$ is indeed a connected component of
$C_{T_0}'$ and $\am_{T_0}(f)$ is represented by $f$ relative to $(T_0,x_0)$.
\item Choose any $(\theta,x)\in C_{T_0,x_0}$ and 
let $\tau_\theta$ be the optimal stopping time as in \eqref{e:optimal_stopping_time}.
Due to  $ X_{\tau_\theta} \in \cl\,\pi(C_{T_0,x_0})$ we have 
\[ v_{\am,g}(\theta,x) = \E_x(e^{-r\tau_\theta}g(X_{\tau_\theta}))
= \E_x(e^{-r\tau_\theta}\widetilde{g}(X_{\tau_\theta}))
\leq  v_{\am, \widetilde{g} }(\theta,x).\]
The reverse inequality follows immediately from the assumption $\widetilde{g} \leq g$.
Therefore
\[\widetilde{g}(x) = g(x) < v_{\am, g}(\theta,x) =  v_{\am, \widetilde{g} }(\theta,x).\]
This shows that $C_{T_0,x_0}$ is a connected subset of $\widetilde{C}_{T_0}$.

Now choose any boundary point $(\theta,x) \in \partial C_{T_0,x_0}$ and an approximating sequence $(\theta_n,x_n)_{n\in \nn}$ in $C_{T_0,x_0}$, i.e.\ 
$(\theta_n,x_n) \to (\theta,x)$ as $n \to \infty$.
We have $g(x) = \widetilde{g}(x)$. Since $ v_{\am,\widetilde g}(\theta_n,x_n) = v_{\am, g}(\theta_n,x_n)$ for any $n\in \nn$,
we conclude
\begin{equation*}
v_{\am,\widetilde g}(\theta,x) 
 = \lim_{n\to\infty} v_{\am,\widetilde g}(\theta_n,x_n)
 =  \lim_{n\to\infty} v_{\am, g}(\theta_n,x_n)
      = g(x) = \widetilde{g}(x).
\end{equation*}
Consequently $(\theta,x)$ is located in the stopping region of the American payoff $\widetilde{g}$.
Summing up,  $C_{T_0,x_0}$  is a connected component of the set $\widetilde{C}_{T_0}$ and 
$\widetilde g$ is represented by $f$ relative to $(T_0,x_0)$.\qedhere
\end{enumerate}
\end{proof}

A European payoff often -- but not always -- generates its embedded American option: 
\begin{proposition}\label{p:EAO}
Suppose that $f$ is continuous.
Let $T_0\in(0,T]$ and assume that there exists a continuous function $\ctheta:\rr\to[0,T_0]$ such that the 
the infimum in the definition of $\am_{T_0} (f)$, cf.\ (\ref{e:EAO}), is reached in $\ctheta(x)$ for any $x\in\rr$.
Then we have:
\begin{enumerate}
\item
$\am_{T_0} (f)$ is continuous.
\item 
$f \represents{T_0,x_0} \am_{T_0}(f)$
and hence 
\[f= \eu_{T_0,x_0} (\am_{T_0} (f))\]
for any $x_0\in\rr$ with $(T_0,x_0)\in C_T$.
\item $\ctheta(x)$ corresponds to the early exercise curve of $\am_{T_0}(f)$
in the sense that
\begin{equation}\label{e:taubla}
 \tau:=\inf\bigl\{t\geq0:T_0-t=\ctheta(X_t)\bigr\} \wedge T_0
\end{equation}
is an optimal stopping time for the stopping problem in the definition of $v_{\am,\am_{T_0}(f)}(T_0,x)$, cf.\ (\ref{e:vam}).
\item $g:=\am_{T_0}(f)(x)$ satisfies the concavity condition
\begin{equation}\label{e:concavity}
{\sigma^2\over2}g''(x)+\left(r-{\sigma^2\over2}\right)g'(x)-rg(x)\leq0
\end{equation}
on the set
\[G:=\{x\in\rr:0<\ctheta(x)<T_0  \text{ and }  g\text{ is twice differentiable in }x\}.\]
\end{enumerate}
\end{proposition}
\begin{proof}
\begin{enumerate}
\item $v_{eu,f}$ is continuous on $(0,T_0)\times\rp$ by Lemma \ref{LEMMA_analyticity_sol_heat_eaqtion}. 
The integrability condition $v_{eu,f}(T,x)<\infty$, $x\in\rr$ and dominated convergence yield that continuity actually holds on $(0,T_0]\times\rp$.
Since $f$ is uniformly integrable relative to $P_x^{X_{\theta}}$ for $(\theta,x)\in[0,T_0]\times[x-\epsilon,x+\epsilon]$
and since $P_x^{X_{\theta}}\to\delta_{x_0}$ weakly for $(\theta,x)\to(0,x_0)$,
the function $v_{eu,f}$ is in fact continuous on $[0,T_0]\times\rr$.
Since $\ctheta$ is continuous, we have that 
$x\mapsto\am_{T_0}(f)(x)=v_{eu,f}(\ctheta(x),x)$ is continuous as well.
\item
$f$ superreplicates the payoff $ \am_{T_0}(f)$ up to $T_0$ by definition.
Since $\theta\mapsto v_{am,g}(\theta,x)$ is increasing,
$(\theta,x)\in C_{T_0,x_0}$ implies $(\widetilde\theta,x)\in C_{T_0,x_0}$ for any $\widetilde\theta\geq\theta$.
Now $\am_{T_0}(f)(x)=v_{eu,am_{T_0}(f)}(\ctheta(x),x)\geq v_{am,am_{T_0}(f)}(\ctheta(x),x)$
implies $(\ctheta(x),x)\notin C_{T_0,x_0}$ and therefore
$\theta>\ctheta(x)$ 
for any $(\theta,x)\in C_{T_0,x_0}$.
Set $M:=\bigl\{(\theta,x)\in[0,T_0]\times\rr: \am_{T_0}(f)(x)=v_{\eu,f}(\theta,x)\bigr\}$.
Since $\ctheta$ is continuous, this implies that 
$P_x((\theta-\tau_{\theta,M},X_{\tau_{\theta,M}})\in M)=1$
for any $(\theta,x)\in C_{T_0,x_0}$ and the stopping time $\tau_{\theta,M}:=\inf\{t\in\rp:(\theta-t,X_t)\in M\}$.
Since the discounted European value process is a martingale, we obtain
\begin{align*}
 v_{am,am_{T_0}(f)}(\theta,x)
 &\geq E_x\bigl(e^{-r\tau_{\theta,M}}g(X_{\tau_{\theta,M}})\bigr)\\
 &=E_x\bigl(e^{-r\tau_{\theta,M}}v_{\eu,f}(\theta-\tau_{\theta,M},X_{\tau_{\theta,M}})\bigr)\\
&=v_{\eu,f}(\theta,x)
\end{align*}
by optional sampling.
The reverse inequality is \eqref{e:glessf} from Proposition \ref{p:PROPOSITION_1}.
\item
This follows now from
\begin{align*}
v_{am,am_{T_0}(f)}(T_0,x)
&\geq E_x\bigl(e^{-r\tau}\am_{T_0}(f)(X_{\tau})\bigr)\\
&=E_x\bigl(e^{-r\tau}v_{\eu,f}(T_0-\tau,X_{\tau})\bigr)\\
&=v_{\eu,f}(T_0,x)\\
&\geq v_{am,am_{T_0}(f)}(T_0,x).
\end{align*}
\item The mapping $\Psi : v_{\eu,f} - g $ is twice differentiable on the set $(0,T_0)\times G$.
If we define the operator  $ \mathscr A  := \Bigl( r -\frac{\sigma^2}{2} \Bigr) D_2  + \frac{\sigma^2}{2} D_{22} - r   $,
It\=o's formula and the martingale property of $(e^{-rt}v_{\eu,f}(T-t,X_t))_{t\in[0,T_0]}$ yield
that $(\mathscr{A}-D_1)v_{\eu,f}=0$ on $(0,T_0)\times\rr$ and hence
\begin{align}\label{EQ_A_applied_to_g}
   \mathscr{A} g 
   &= \mathscr{A} v_{\eu,f} -  \mathscr{A} \Psi \nonumber \\
  &= D_1 v_{\eu,f} - \left( r - \frac{\sigma^2}{2} \right) D_2 \Psi
     - \frac{\sigma^2}{2} D_{22} \Psi - r \Psi \nonumber\\
  &= c^\top D \Psi - \frac{\sigma^2}{2} D_{22} \Psi - r \Psi 
\end{align} 
where $c:= (1,  \frac{\sigma^2}{2} -r )$ and $g$ is interpreted as a mapping
on $(0,T_0)\times\rr$ via $g(t,x):=g(x)$. Now choose any $x \in G$. 
By definition we have $\Psi(\ctheta(x), x) = 0$. Due to the fact that $\Psi$ only assumes nonnegative values,
the first order condition $D\Psi(\ctheta(x), x)  = 0$ and the second order
condition $ D_{22} \Psi(\ctheta(x), x)\geq 0$ hold.
From \eqref{EQ_A_applied_to_g} we obtain
 \begin{align*}
  (\mathscr{A} g)(x) =- \frac{\sigma^2}{2}D_{22} \Psi(\ctheta(x), x) \leq 0,
 \end{align*}
which concludes the proof.
\qedhere
\end{enumerate}
\end{proof}

\subsection{Examples}
We start with a simple explicit example of a representable American option.
\begin{example} Consider the market of Section \ref{s:SEC_intro} with
interest rate $r=0$ and volatility $\sigma=\sqrt{2}$.
We study the European payoff
\[f(x) = 3 e^{x/2} +  e^{3x/2}.\]
Since 
\[\E_x (S_t^\alpha) = \exp\bigl(\alpha x+(\alpha^2 - \alpha)t\bigr)\]
for $\alpha \in \rr$, $x > 0$, 
its value function equals
\[v_{\eu,f}(\theta,x) =3\exp\biggl({1\over2}x-{1\over4}\theta\biggr)+\exp\biggl({3\over2}x+{3\over4}\theta\biggr).\]
We conclude that 
the embedded American option and the associated early exercise curve are given by
\begin{align*}
\am_{\infty}(f)(x) &= 4 e^{3x/4} 1_{(-\infty,0)}(x) + f(x) 1_{\rp}(x),\\
\ctheta(x) &= \mathrm{argmin}_{\theta \in \rp } v_{\eu,f}(\theta,x) = -x 1_{(-\infty,0)}(x).
\end{align*}
More specifically, $\tau$ in \eqref{e:taubla} is optimal for the stopping problem \eqref{e:vam}
for $g=\am_{\infty}(f)$ and time horizon $T_0$.
Indeed, $\tau$ is optimal for 
\[\am_{T_0}(f)(x)=
    \begin{cases}
						3 e^{x/2-T_0/4} + e^{3x/2-3T_0/4} \ & \text{ if } x \leq-T_0,\\
						\am_{\infty}(f)(x) & \text{ otherwise.}
    \end{cases}\] 
    by Proposition \ref{p:EAO}.
Since $\am_{\infty}(f)\leq\am_{T_0}(f)$, it follows easily that it is optimal for $\am_{\infty}(f)$ as well.
\end{example}

The following example shows that the embedded American option $\am_{T}(f)$ of $f$ may be representable without necessarily being represented by $f$ itself.
On top, we observe that the early exercise curve can have jumps.
\begin{example}\label{Example1}\index{representability (function-type)!examples!locally but not globally rep., discontinuous exercise boundary}
Consider the Black-Scholes market of Section \ref{s:SEC_intro} with interest rate $r=1$ and volatility $\sigma = \sqrt{2}$.
The European value function associated to the payoff $f:=1_{[0,1]}$ is given by
\[v_{\eu,f}(\theta,x) = e^{-\theta} \left( \Phi\!\left( \frac{1-x}{\sqrt{2\theta}} \right) - \Phi\!\left( -\frac{x}{\sqrt{2\theta}} \right)  \right).\]
An elementary calculation yields
\begin{align}\label{e:less0}
D_1 v_{\eu,f}(\theta,x) 
&= - e^{-\theta} (2\theta)^{-3/2}  \left( \varphi\!\left( \frac{1-x}{\sqrt{2\theta}} \right)(1-x) + \varphi\!\left( \frac{x}{\sqrt{2\theta}}  \right) x\right)
- v_{\eu,f}(\theta,x)< 0 
\end{align}
for any $(\theta,x) \in (0,\infty)\times [0,1]$. 
Moreover, we have
\begin{equation}\label{e:EQ_Example1_cases}
\lim_{\theta \downarrow 0} v_{\eu,f}(\theta,x)  = \begin{cases}
						\frac{1}{2}<f(x) & \text{ if } x \in \{0,1\},\\
						f(x) & \text{ otherwise.}
                                              \end{cases}
\end{equation}

Fix some time horizon $T\in(0,\infty)$.
In view of (\ref{e:less0}, \ref{e:EQ_Example1_cases}), the embedded American option is given by
\[g(x):= \am_{T}(f)(x) = v_{\eu,f}(T,x) 1_{[0,1]}(x).\]
The infimum in \eqref{e:EAO} is attained at the unique point
\[ \ctheta(x) = T  1_{[0,1]}(x),\quad x\in\rr.\]
This shows that neither the embedded American option nor the associated curve $x\mapsto\ctheta(x)$ of unique minima need to be continuous if 
the underlying European payoff is discontinuous. 
The reader may compare this result to the statements of Proposition \ref{p:EAO} and \cite[Remark 4]{jourdain.martini.01}.

Let $C_T$ denote the continuation region as \eqref{e:cont}. Since 
$g(x)=0<v_{\am,g}(\theta ,x)$
for $(\theta,x)\in(0,T]\times(\rr\setminus[0,1])$, it is evident that $(0,T]\times(\rr\setminus[0,1])\subset C_T$.
For $(\theta,x)\in[0,T]\times[0,1]$ we have 
\[g(x)\leq v_{\am,g}(\theta ,x)\leq v_{\am,g}(T ,x) \leq v_{\eu,f}(T ,x)=g(x)\]
because the American value function $v_{\am,g}(\theta ,x)$ is increasing in $\theta$.
Consequently, the continuation region is of the form $C_T=(0,T]\times(\rr\setminus[0,1])$.
At the end of this example we prove that the stopping time 
$\tau_\theta:=\inf\{t\geq0:X_t\in[0,1]\}\wedge\theta$  
is optimal for the stopping problem \eqref{e:vam} with time horizon $\theta\in[0,T]$.

Beforehand we show that the embedded American payoff $g$ is \emph{not} represented by its generating European claim $f$. 
To this end choose any $(\theta ,x)\in[0,T]\times\rp$ from the continuation region $C_T$.
Since 
$ D_1 v_{\eu,f}(\theta,x)<0$ 
on the set $(0,\infty) \times [0,1]$
and by the optional sampling theorem applied to the discounted European option price process,
we obtain
\begin{align*}
 v_{\am,g}(\theta,x) 
 &= \E_x\!\left(  g(X_{\tau_\theta})    e^{-r\tau_\theta}  \right)\\
 &= \E_x\!\left( 1_{[0,1]}(X_{\tau_\theta}) v_{\eu,f}(T,X_{\tau_\theta})  e^{-r\tau_\theta}  \right)\\
 &<  \E_x\!\left(  v_{\eu,f}(\theta -\tau_\theta,X_{\tau_\theta})  e^{-r\tau_\theta}  \right)\\
 &= v_{\eu,f}(\theta ,x)
\end{align*}
for any  $(\theta ,x) \in C$ with $\theta\leq T$.
Therefore the payoff $g$ is indeed not represented by $f$.

Nonetheless, there exist European payoff functions which represent $g$ on the connected components 
$C_{T,-1}= (0,T] \times (-\infty,0)$ and  $C_{T,2}= (0,T] \times (1, \infty)$ of the continuation region. 
We verify that 
\[h(x) := 2 g(0)\cosh(x)1_{\rr_+}(x)\] 
represents $g$ on the left connected component $C_{T,-1}$.
By symmetry, one can show that the same holds for 
$\widetilde{h}(x) := 2 g(1)\cosh(x-1)1_{(-\infty,1]}(x)$ on the right connected component $C_{T,2}$.

The European value function associated to $h$ is given by
\[ v_{\eu,h}(\theta,x) =  2 g(0)e^{-\theta}H(\theta,x), \]
where $H(\theta,x):= \E_x(\cosh(X_\theta)1_{\rr_+}(X_\theta))$.
Since $P_x^{X_\theta}= N(x,2\theta)$, a straightforward calculation yields 
\[H(\theta,x)={1\over2}\left(e^{\theta-x}\Phi\biggl({x\over\sqrt{2\theta}}-\sqrt{2\theta}\biggr)+e^{\theta+x}\Phi\biggl({x\over\sqrt{2\theta}}+\sqrt{2\theta}\biggr)\right)\]
for $\theta\in(0,T]$. In particular, we have
\begin{equation}\label{e:innull}
  v_{\eu,h}(\theta,0) =  g(0). 
\end{equation}

Let us verify that $h$ superreplicates the American payoff $g$ up to time $T$.
Since
\begin{align*}
g(x)
&=e^{-T} \left( \Phi\!\left( \frac{1-x}{\sqrt{2T}} \right) -1+ \Phi\!\left( \frac{x}{\sqrt{2T}} \right)  \right)\\
&\leq e^{-T} \left( \Phi\!\left( \frac{1}{\sqrt{2T}} \right) -1+ \Phi\!\left( \frac{1}{\sqrt{2T}} \right)  \right)\\
&= 2g(0)\\ 
&\leq h(x) = v_{\eu,h}(0,x)
\end{align*}
for any $x\in[0,1]$, it suffices to verify $v_{\eu,h} \geq g$ on the set $(0,T]\times [0,1]$.
In view of
\begin{align*}
D_1  v_{\eu,h}(\theta,x)
=- g(0){2xe^{-\theta}\over(2\theta)^{3/2}}\phi\biggl({x\over\sqrt{2\theta}}\biggr)<0 
\end{align*}
for any $(\theta,x)\in(0,T]\times [0,1]$,
we only need to show that
$v_{\eu,h}(T,x) - g(x)=v_{\eu,h-f}(T,x)$ is nonnegative for any $x \in [0,1]$.
We have
\[v_{\eu,h-f}(T,x)=E_x((h-f)(X_T)|X_T\geq0)P_x(X_T\geq0)\]
because $h-f$ vanishes on $(-\infty,0)$.
Since $h-f$ is increasing, Lemma \ref{l:stochasticdominance} yields 
\[{v_{\eu,h-f}(T,0)\over P_0(X_T\geq0)}\leq {v_{\eu,h-f}(T,x)\over P_x(X_T\geq0)}\] 
for any $x\in[0,1]$.
Using \eqref{e:innull} we conclude
\[0=v_{\eu,h}(T,0)-v_{\eu,f}(T,0)= v_{\eu,h-f}(T,0),\quad x\in[0,1]\] 
and hence $0\leq v_{\eu,h-f}(T,x)$
as desired.

We already observed that the functions $v_{\eu,h}$ and $g$ coincide on the stopping boundary associated to $C_{T,-1}$, i.e.\ 
$v_{\eu,h}(0,x) = 0 = g(x)$ for any $x<0$ and $v_{\eu,h}(\theta,0) = g(0)$ for any $\theta \in [0,T]$.
Consequently, optional sampling yields
\begin{align*}
 v_{\am,g}(\theta,x) 
 &= \E_x\bigl( g(X_{\tau_\theta}) e^{-r\tau_\theta}\bigr)\\
 &= \E_x\bigl( v_{\eu,h}(\theta -\tau_\theta,X_{\tau_\theta}) e^{-r\tau_\theta}\bigr)\\
 &= v_{\eu,h}(\theta ,x)
\end{align*}
for any $(\theta,x)\in C_{T,-1}$.
In particular, we observe that $\tau_\theta$ is an optimal stopping time for the stopping problem \eqref{e:vam} with time horizon $\theta$.
Altogether, this shows that the American payoff $g$ is represented by $h$ on the left connected component $C_{T,-1}$.
\end{example}

\subsection{The American option embedded into the European put}\label{su:subsection_am_opt_generated_by_put}
\begin{figure}
\centering
\includegraphics[width=0.85\textwidth]{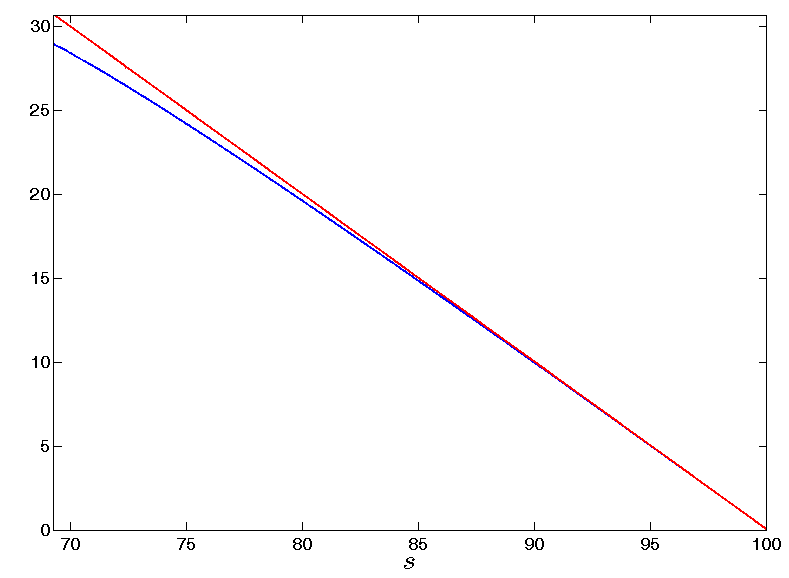}
\caption{The EAO (blue) associated to the European put with strike $K=100$ (red) in the Black-Scholes market with $T=1$, $r=0.06$, $\sigma=0.4$, and stock price $s=e^x$}\label{f:euput1}
\end{figure}
\begin{figure}
\centering
\includegraphics[width=0.85\textwidth]{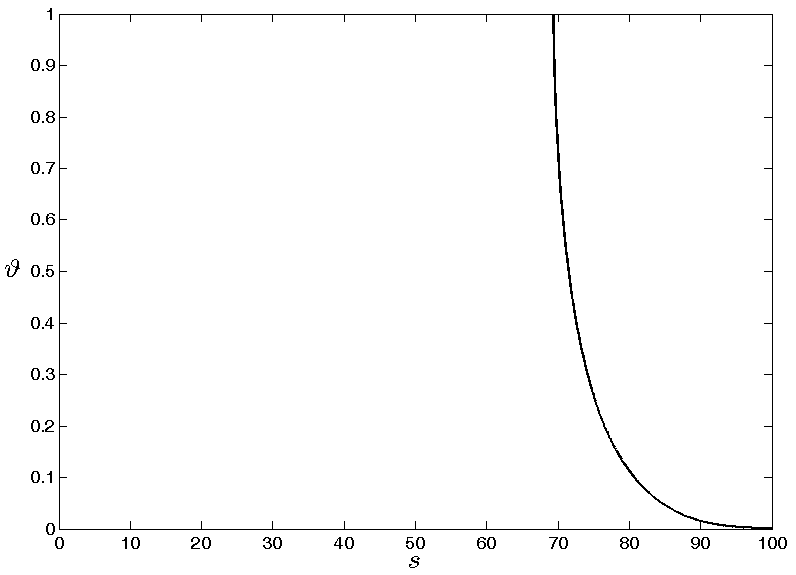}
\caption{The curve $\ctheta(x)$ associated to the European put with strike $K=100$ in the Black-Scholes market with $T=1$, $r=0.06$, $\sigma=0.4$, and stock price $s=e^x$}\label{f:euput2}
\end{figure}
The embedded American option of the European put has some interesting properties.
It is representable, but only for sufficiently small time horizons.
\begin{lemma}
By $f(x):=(K-e^x)^+$ we denote the European put option with some strike price $K>0$.
There are positive finite time horizons $T_1<T_2$ such that
\begin{enumerate}
 \item $\am_{T_1}(f)$ is represented by $f$ relative to $(T_1,x)$ for all $x\in[\log K,\infty)$,
 \item If $r>0$, there are no $T\in\rp$, $x\in[\log K,\infty)$ such that $\am_T(f)$ is representable relative to $(T_2,x)$.
 In particular, neither $\am_{T_1}(f)$ nor $\am_{T_2}(f)$ are represented by $f$ relative to $(T_2,\log K)$.
\end{enumerate}
\end{lemma}
\begin{proof}
\begin{enumerate}
\item Owing to the Black-Scholes formula, the value function of the European put is given by
\begin{equation}\label{EQ_Black_scholes_formula}
v_{\eu, f}(\theta, x) =   e^{-r \theta} K\Phi(-d_2(\theta,e^x)) -  s \Phi(-d_1(\theta,e^x))
\end{equation}
with
\begin{equation*}
\begin{aligned}
d_1(\theta,s) &:= \frac{\log(s/K) + (r  + \sigma^2/2)\theta}{\sigma\sqrt{\theta}}, \\
d_2(\theta,s) &:= \frac{\log(s/K) + (r  - \sigma^2/2)\theta}{\sigma\sqrt{\theta}}, 
\end{aligned}
\end{equation*}
where $s=e^x$ denotes the spot price of the underlying and $\theta \in \rr_+$ the maturity of the option.
We show that for sufficiently small terminal time $T$, there exists a continuous function $\ctheta(x):\rr\to[0,T]$ with
$\am_{T}(f)(x) = v_{\eu, f}(\ctheta(x),x)$ for any $x\in\rr$. 
Proposition \ref{p:EAO} then warrants that the American payoff $\am_{T}(f)$ is represented by $f$ relative to any $x\in\rr$ with $(T,x)\in C_T$.
We recall the following well-known partial derivatives of $v_{\eu, f}$:
\begin{align*}
D_1v_{\eu, f} (\theta,x)&=   \frac{e^x \sigma}{2 \sqrt{\theta}} \varphi(d_1(\theta,e^x)) - r K e^{-r \theta}\Phi(-d_2(\theta,e^x)), \\
e^{-x}D_2v_{\eu, f}(\theta,x) &=   - \Phi(-d_1(\theta,e^x)), \\
e^{-x}D_{12}v_{\eu, f}(\theta,x) &=   \frac{ \left( r +  \sigma^2/2 \right) \theta - \log\!\left( e^x/K \right) }{2 \theta^{3/2} \sigma}  \varphi(d_1(\theta,e^x)).
\end{align*}
Consequently, for any $(\theta,x)\in \rr_+ \times \rr$  we have 
\begin{equation}\label{EQ_dxTheta_positive}
D_{12}v_{\eu, f}(\theta,x) > 0
\end{equation}
if and only if $e^x < K \exp\!\left( (r + \sigma^2 /2)\theta \right)$.
Moreover, one easily verifies that the following properties are satisfied for any $T>0$:
\begin{align}
&\liminf_{x\downarrow -\infty} \sup_{\theta\in [0,T]} D_1v_{\eu, f}(\theta,x) < 0 \label{EQ_Theta_neg_at_small_price},\\
&\lim_{\theta \downarrow 0} D_1v_{\eu, f}(\theta,\log K) =  \infty \label{EQ_Theta_large_at_strike},  \\
&\lim_{\theta\downarrow 0} D_1v_{\eu, f}(\theta,x) =  -r K,\quad x \in (-\infty, \log K)  . \label{EQ_Theta_neg_at_beginning} 
\end{align}
By \eqref{EQ_Theta_large_at_strike} there is some constant $T_{\max}>0$ such that $D_1v_{\eu, f}(\theta,\log K) > 0$ for any $\theta\in (0,T_{\max})$.
Let $T\in(0,T_{\max})$.
Property \eqref{EQ_Theta_neg_at_small_price} warrants that $\liminf_{x\downarrow-\infty} D_1v_{\eu, f}(T,x) < 0$.
Due to \eqref{EQ_dxTheta_positive} and the intermediate value theorem, there exists a unique $K_T \in (0,K)$
such that $D_1v_{\eu, f}(T,\log K_T)=0$, $D_1v_{\eu, f}(T,x) < 0$ for $x\in (0,\log K_T)$, and $D_1v_{\eu, f}(T,x) > 0$ for $x\in (\log K_T,\log K]$. 
Taking \eqref{EQ_Theta_neg_at_beginning} into account, we conclude that
\[m(x):=\min_{\theta \in [0,T]} v_{\eu, f}(\theta,x) < v_{\eu, f}(0,x) \wedge v_{\eu, f}(T,x),\quad x\in (\log K_T,\log K).\]
Put differently, the nonempty compact set 
\[M_x := \{ \theta \in [0,T] : v_{\eu, f}(\theta,x) = m(x)\}\] 
is contained in the open interval $(0,T)$ for any $x \in (\log K_T,\log K)$. We write \[\ctheta(x):= \max M_x\] for the largest value of the set $M_x$.
For any $x \in (\log K_T,\log K)$ we have $D_1v_{\eu, f}(\ctheta(x),x)=0$. By decreasing the bound $T_{\max}$ we can always achieve that 
$D_{11}v_{\eu, f}(\ctheta(x),x)>0$ for any $x \in (\log K_T,\log K)$. This can be verified by analysing the asymptotic behaviour of the 
derivative $D_{11}v_{\eu, f}$ as $\theta\to 0$. The calculation is elementary but somewhat lengthy and therefore omitted. 
Theorem \ref{THEOREM_Implicit_Function} yields that the mapping $x \mapsto \ctheta(x)$ is analytic on some open complex domain containing the interval $(\log K_T,\log K)$.
Moreover, owing to \eqref{EQ_dxTheta_positive} we have
\begin{equation}\label{EQ_derivative_of_curve_ctheta}
\ctheta'(x) = - \frac{  D_{12}v_{\eu, f}(\ctheta(x),x)}{D_{11}v_{\eu, f}(\ctheta(x),x)} < 0  
\end{equation}
for any $x \in (\log K_T,\log K)$, which implies that the limits $\lim_{x \downarrow \log K_T} \ctheta(x)$ and  $\lim_{x \uparrow \log K} \ctheta(x)$ exist.
Note that the mapping $[0,T] \ni \theta \mapsto v_{\eu, f}(\theta,x)$ attains its unique minimum at $\theta=0$ for any $x\geq \log K$.

A simple calculation shows that $\ctheta(\log K):=0$ extends the curve $\ctheta$ continuously to $x=\log K$.
Indeed, assuming $v_1:= \lim_{x \uparrow \log K} \ctheta(x) > 0$ yields $\ctheta(x) \in (v_1,T]$ for any $x \in (\log K_T,\log K)$. The mapping 
$D_1v_{\eu, f}$ is continuous on $(0,\infty)\times\rr$. Thus we obtain the contradiction
$ 0<  D_1v_{\eu, f}(v_1,\log K) = \lim_{x \uparrow \log K} D_1v_{\eu, f}(\ctheta(x),x) = 0$.

By possibly decreasing $T_{\max}$ further, we can achieve that $\ctheta(\log K_T) := T$ extends the curve continuously into $x=\log K_T$.

A similar argument as above shows that for any $x\in (0,\log K_T)$ the minimum of the mapping $[0,T] \ni \theta \mapsto v_{\eu, f}(\theta,x)$ is attained at $\theta = T$.
Indeed, by \eqref{EQ_Theta_neg_at_beginning} no minimum can be located at $\theta = 0$.
Now assume that for some $x \in (0,\log K_T)$ a minimum is attained at some maturity $\rho \in (0,T)$. 
Denoting by ${\ctheta}^{-1}$ the inverse function of $\ctheta|_{(\log K_T,\log K)}$,
property \eqref{EQ_dxTheta_positive} yields
$0 = D_1v_{\eu, f}(\rho,x) < D_1v_{\eu, f}(\rho, {\ctheta}^{-1}(\rho)) = 0$
and hence a contradiction.

Altogether we have found the desired function $\ctheta:\rr\to[0,T]$.

\item 
With regard to the Black-Scholes formula \eqref{EQ_Black_scholes_formula}, it is apparent that $\lim_{\theta\to\infty} v_{\eu,f}(\theta,x) = 0$ for any $x\in\rr$.
For fixed $x_0<\log K$ choose $T_2$ large enough with $v_{\eu,f}(T_2,x_0)<\am_T(f)(x_0)$.
Let $T\in\rp$ be arbitrary and $T_1\leq T$ as in the first assertion.

Assume by contradiction that $\am_T(f)$ is represented by some $\widetilde f$ relative to $(T_2,x)$ with some $x\geq \log K$.

For sufficiently large $x_0<\log K$ we have
$\am_T(f)=\am_{T_1}(f)$ on $[x_0,\infty)$.
Indeed, $\am_{T_1}(f)(x)\to 0$ as $x\to 0$ and 
$x\mapsto \inf_{\theta\in[T,T_1]}v_{\eu,f}(\theta,x)$ has a positive lower bound on the compact interval $[\log K-1,\log K]$ because
$v_{\eu,f}$ is continuous and strictly positive on $[T_1,T]\times[\log K-1,\log K]$.

After possibly decreasing $T_1$ we can apply Proposition \ref{PROPOSITION_2}(6) and obtain that $f$ represents $\am_T(f)$ 
relative to $(T_1,x)$.
Proposition \ref{PROPOSITION_2}(1,2) yields that the mappings $f$ and $\widetilde{f}$ coincide up to a Lebesgue-null set. 
Hence we obtain the contradiction
\[\am_T(f)(x_0) \leq v_{\eu,\widetilde{f}}(T_2,x_0) = v_{\eu,f}(T_2,x_0)< \am_T(f)(x_0).\qedhere\]
\end{enumerate}
\end{proof}

The embedded American option of the European put and the curve $\ctheta$ in the proof of the previous lemma are illustrated in Figures \ref{f:euput1}, \ref{f:euput2}.

\section{Existence of the CDEO and representability}\label{s:section_verification_theorem}
The aim of this section  is to establish the existence of cheapest dominating European options and, more importantly, to verify that that a given American option 
is represented by its CDEO.
For ease of exposition we focus on payoffs of a particular form. 

\subsection{Main results}
We consider the basic model of (\ref{e:EQ_BS_market_from_Intro}) 
with initial logarithmic stock price $X_0=x_0$ and fixed time horizon $T$.
We are primarily interested in the American put but for the theorems below it satisfies to assume a certain more general structure.
Specifically, we consider payoffs of the form
\begin{align}\label{e:TYPE_ONE_PAYOFF}
 g(x) = 1_{(-\infty,K]}(x) \phi(x),\quad x\in\rr
\end{align}
with $K\in\rr$ and an analytic function $\phi:U\to\cc$ on some domain $U\subset\cc$ 
such that
\begin{enumerate}
 \item $\phi(x)\in(0,\infty)$ for $x\in(-\infty,K)$,
 \item $\phi(K) = 0$,
 \item the growth condition
 \begin{equation}\label{EQ_Growth_condition_payoff}
 \lim_{\rr\ni x\to -\infty} e^{ (2r/\sigma^2)x}\phi(x) = 0
 \end{equation}
 holds,
 \item the necessary concavity condition
 \begin{align}\label{eq:function_c}
 c(x):=g''(x)-  \frac{2r}{\sigma^2}(g(x)-g'(x))- g'(x) \leq 0,\quad x\in(-\infty,K)
 \end{align}
 from \eqref{e:concavity} in Proposition \ref{p:EAO} holds.
\end{enumerate}
These assumptions are satisfied and in fact motivated by the payoff $g(x)=\left( e^K - e^x \right)^+$, 
which corresponds to the American put.

Our first goal is to  show that the cheapest dominating European option of $g$ relative to $(T,x_0)$ exists
in a suitably generalised sense. 
If $f:\rr\to\rp$ denotes a European payoff function, we have
\begin{equation}\label{eq:Def_Veu_f_without_normalisation}
v_{\eu, f}(\theta, x)=
e^{-r\theta}\int\nno\bigl(x+ \widehat r\theta,\sigma^2\theta,y\bigr) f(y)d y,
\end{equation}
where 
$\widehat r:=r-{\sigma^2\over2}$.
Put differently, we obtain
\begin{equation}\label{e:mupreis}
v_{\eu, f}(\theta, x)
=e^{-r\theta}\int\frac{ \nno\!\left(x +\widehat r \theta, \sigma^2 \theta, y \right)}{ \nno\!\left(x_0 +\widehat r T, \sigma^2 T, y \right)}\mu(dy)
\end{equation}
for the measure $\mu$ on $\rr$ with density
$f$ relative to $\nnormal(x_0 +\widehat r T, \sigma^2 T)$.

In the European valuation problem, the payoff function $f$ is only needed for defining the pricing function $v_{\eu, f}$.
In view of (\ref{e:mupreis}) we can and do therefore extend the notion of a payoff ``function'' to include all 
$\mu\in\M^+ (\rr)$, where $\M^+(\rr)$ denotes the set of measures on $\rr$. 
In line with (\ref{e:mupreis}), we define the pricing operator $v_{\eu, \mu}:\rp\times\rr\to[0,\infty]$ by
\begin{equation}\label{e:mu_preis_op}
 v_{\eu, \mu}(\theta, x)
:=e^{-r\theta}\int \frac{ \nno\!\left(x +\widehat r \theta, \sigma^2 \theta, y \right)}{ \nno\!\left(x_0 +\widehat r T, \sigma^2 T, y \right)}\mu(dy), 
\quad (\theta,x)\in (0,\infty)\times\rr
\end{equation}
and
\begin{equation}\label{e:mu_preis_op_rand}
 v_{\eu, \mu}(0, x):=\liminf_{(\theta,y) \to (0,x)}v_{\eu, \mu}(\theta, y),\quad x\in\rr.
\end{equation}
In terms of our generalised domain, the linear problem (\ref{e:CDEO_general_opt_problem}) now reads as
\begin{equation}
\label{PROGRAM_primal_general}
\begin{aligned}
& \text{minimise}   & &  v_{\eu, \mu}(T, x_0)  \\
& \text{subject to} & &   \mu \in \M^+(\rr), \\
&&& v_{\eu, \mu}(\theta,x)\geq g(x)\text{ for any } (\theta,x)\in [0,T]\times\rr.
\end{aligned}
\end{equation}
In line with Definition \ref{d:DEF_EAO_CDEO}(3), 
a minimiser $\mu^*$ is called \emph{cheapest dominating European option (CDEO)} of $g$ relative to
$(T,x_0)$. 
Our first main result establishes its existence.
\begin{theorem}\label{t:Theorem_summary0} The optimal value of programme \eqref{PROGRAM_primal_general} 
is obtained by some $\mu^* \in \M^+(-\infty,K]$, i.e.\ some measure on $\rr$ which is concentrated on $ (-\infty,K]$.
In particular, a CDEO of $g$ relative to $(T,x_0)$ exists in the present generalised sense.
\end{theorem}
The proof is to be found in Section \ref{Section_Proof_Thm1}.
We now turn to the question whether the CDEO actually generates the American claim $g$ under consideration.

\begin{theorem}\label{t:THEOREM_Hauptsatz} Let $\mu^*$ denote an optimal measure from Theorem \ref{t:Theorem_summary0}. 
Suppose that the following assumptions are satisfied for some constant $\delta >0$:
\begin{enumerate}
\item There exists some  $x_1 \in \rr$ such that  $v_{\eu, \mu^*}(T+2\delta, x_1)<\infty$. 
\item For any $\theta\in (0,T+\delta)$ the function $x \mapsto v_{\eu, \mu^*}(\theta, x)-g(x)$ 
assumes its unique minimum within the interval $(-\infty,K]$ at some point $\cx(\theta) \in (-\infty,K)$. 
Moreover, we have $\liminf_{\theta\to 0} \cx(\theta) = K =: x(0)$.
\item The well-defined quantity
\begin{equation}\label{EQ_function_H_of_verification_theorem}
H(\theta,x):= \frac{2}{\sigma^2}D_1 v_{\eu, \mu^*}(\theta, x) 
+ \frac{2r}{\sigma^2}(v_{\eu, \mu^*}(\theta, x)-g(x))-c(x)
\end{equation}
is strictly positive on the set 
$\{(\theta, \cx(\theta)) : \theta\in(0,T]\}$. 
\item We have $\liminf_{\theta \to 0} v_{\eu, \mu^*}(\theta, K) <\infty$. 
\end{enumerate}
Define $C_{T,x_0}$ as in Section \ref{s:EAO_CDEO} and set
\begin{equation}\label{EQ_set_C_tilde_from_verification_theorem}
\widetilde C_{T,x_0}:= \{(\theta,x) \in(0,T] \times \rr: \cx(\theta) < x \}.
\end{equation}
Then the following statements hold.
\begin{enumerate}
 \item The function $\theta \mapsto \cx(\theta)$ 
 is strictly increasing and it can be extended to some analytic function on a complex domain containing $(0,T]$.
 \item We have $v_{\eu, \mu^*}(\theta,\cx(\theta))  = g(\cx(\theta))$ for any $\theta \in [0,T]$. 
 \item The CDEO $\mu^*$ is the unique measure that \emph{represents} $g$ relative to $(T,x_0)$
 in the sense that $v_{\am, g}(\theta,x)\leq v_{\eu, \mu^*}(\theta,x)$ for any  $(\theta,x)\in[0,T]\times\rr$
 and equality holds on $C_{T,x_0}$. 
 \item The payoff $g$ coincides on $\cl\,\pi( \widetilde C_{T,x_0} ) = [\min_{\theta\in (0,T]} \cx(\theta),\infty)$ 
 with the \emph{embedded American option} of $\mu^*$ up to $T$ 
 in the sense that 
 \[g(x)=\inf_{\theta\in[0,T]} v_{\eu, \mu^*}(\theta,x)=:\am_T(\mu^*)(x) ,\quad  x \in \pi( \widetilde C_{T,x_0}).\]
 \item $\widetilde C_{T,x_0} = C_{T,x_0}$, which can be interpreted in the sense that $\cx$ parametrises the optimal stopping boundary.
 \item The stopping time 
  \begin{equation}\label{EQ_optimal_stopping_time_from_Hauptsatz}
   \tau_\theta :=\inf\{t\in[0,\theta] :  X_t \leq \cx(\theta-t) \} \wedge \theta
  \end{equation}
  is optimal in \eqref{e:vam},  i.e.\ $v_{\am,g}(\theta,x) = \E_x(e^{-r \tau_\theta} g(X_{\tau_\theta}))$ holds for any  $(\theta,x)\in [0,T] \times \rr$.
\end{enumerate}
\end{theorem}
The proof of this theorem is to be found in Section \ref{section_european_representation}.

What are the strengths and weaknesses of this result?
The assumptions above concern certain qualitative properties of 
the cheapest dominating European option. 
On the negative side, this means that Theorem \ref{t:THEOREM_Hauptsatz} does not warrant 
representability unless one can prove that these properties hold 
for the CDEO of the specific claim under consideration. 
This is complicated by the fact that this CDEO is typically not known explicitly.
However, numerical approximations are obtained quite easily as it is explained in \cite[Chapter 3]{lenga.17}.
While such approximations cannot tell whether the CDEO represents the American claim or
just provides a relatively close upper bound, they provide  evidence
whether the \emph{qualitative} properties needed for Theorem \ref{t:THEOREM_Hauptsatz} hold true.
As an illustration, we study the prime example of the American put in Section \ref{s:SEC_Example_am_put}.

\subsection{Proof of Theorem \ref{t:Theorem_summary0}}\label{Section_Proof_Thm1}
First we verify that in programme \eqref{PROGRAM_primal_general} it suffices to consider measures $\mu \in \M^+(-\infty,K]$.
To this end we define by
\begin{align*}
&\M^+(\rr) \ni \mu \mapsto s(\mu) := \nu_1 + \nu_2\\
&d\nu_1 := 1_{(-\infty,K]} d\mu  \\
&d\nu_2 := \mu((K,\infty))  d\delta_{K}
\end{align*}
the mapping which relocates any mass in $(K,\infty)$ to $K$.
Here $\delta_{K}$ denotes the Dirac measure at $K$. 
One easily verifies that $s$ maps onto the cone $\M^+(-\infty,K]$ and preserves the total variation, i.e.\ $\Vert s(\mu) \Vert = \Vert \mu \Vert$.
Let $\mu \in \M^+(\rr)$ be admissible in programme \eqref{PROGRAM_primal_general}.
We have 
\[v_{\eu, \mu}(T, x_0) = e^{-rT} \Vert \mu \Vert = e^{-rT} \Vert s(\mu) \Vert = v_{\eu, s(\mu)}(T, x_0).\]
By Lemma \ref{LEMMA_norm_and_quotient_of_normal_pdf}(2) there is some $c(\theta,x) > 0$ such that
\begin{align*}
\frac{ \nno\!\left(x +\widehat{r} \theta, \sigma^2 \theta, y \right)}{ \nno\!\left(x_0 +\widehat{r} T, \sigma^2 T, y \right)}  &=  c(\theta,x) \exp\!\left( - \frac{(y - A(\theta,x) )^2}{2 B(\theta)}    \right)
\end{align*}
 for any $(\theta,x)\in (0,T) \times (-\infty,K)$, where $A(\theta,x) := x_0 + (x - x_0)T / (T-\theta)$ and $B(\theta) := \sigma^2 T\theta / (T-\theta)$. 
Recalling that  $x < K < x_0$, we obtain $A(\theta,x) = x_0 + (x - x_0)T/(T-\theta) < K$ and therefore 
\begin{align*}
v_{\eu, s(\mu)}(\theta, x)  
&= v_{\eu, \nu_1}(\theta, x) + \mu((K,\infty))  v_{\eu, \delta_K}(\theta, x) \\
&= v_{\eu, \nu_1}(\theta, x) + e^{-r\theta} \int_{(K,\infty)}c(\theta,x) \exp\!\left( - \frac{(K - A(\theta,x) )^2}{2 B(\theta)}    \right)    \mu(dy)\\
&\geq  v_{\eu, \nu_1}(\theta, x) + e^{-r\theta}  \int_{(K,\infty)}c(\theta,x) \exp\!\left( - \frac{(y - A(\theta,x) )^2}{2 B(\theta)}    \right)    \mu(dy)\\
&= v_{\eu, \nu_1}(\theta, x) + v_{\eu, \mu - \nu_1}(\theta, x) \\
&= v_{\eu, \mu}(\theta, x) \geq g(x)
\end{align*}
for any $(\theta,x)\in (0,T) \times (-\infty,K)$.
Along the same lines we can apply Lemma \ref{LEMMA_norm_and_quotient_of_normal_pdf}(3) in order to obtain
\begin{align*}
v_{\eu, s(\mu)}(T, x)  
&= v_{\eu, \nu_1}(T, x) + \mu((K,\infty))  v_{\eu, \delta_K}(T, x) \\
&\geq  v_{\eu, \nu_1}(T, x) + v_{\eu, \mu - \nu_1}(T, x) \\
&= v_{\eu, \mu}(T, x) \geq g(x)
\end{align*}
for any $x < K$.
Summing up, the calculations above imply that the inequality $v_{\eu, s(\mu)} \geq g$ holds on the set $(0,T] \times (-\infty,K)$.
Since $g$ is assumed to vanish on $[K,\infty)$, the measure $s(\mu)$ is admissible in programme \eqref{PROGRAM_primal_general}.
Hence, it suffices to consider measures $\mu \in \M^+(-\infty,K]$ in \eqref{PROGRAM_primal_general}.

\subsubsection{Transformation to $r=0$}\label{SubSubSection_Transformation}
Now we transform our model \eqref{e:EQ_BS_market_from_Intro} to a market with constant bond price process, following the approach explained in \cite{lerche.urusov.07}.
To this end, let
\begin{equation}\label{e:EQ_BS_market_from_Intro2}
\begin{aligned}
 \widetilde B_t &= 1,\\
 d\widetilde X_t &= \widetilde{r}  dt + \sigma dW_t 
\end{aligned}
\end{equation}
with $\widetilde{r} := - r - \sigma^2/2 < 0$  and
\begin{equation}\label{EQ_DEFINITION_of_tilde_g}
\widetilde g(x):=e^{(2r/\sigma^2)x}g(x),
\end{equation}
where $W$ denotes a Wiener process and $\widetilde X_0=x$ holds $\P_x$-almost surely.
The growth condition \eqref{EQ_Growth_condition_payoff} warrants that $\widetilde{g}$ is a continuous function vanishing at infinity.
Invoking a measure change with density process 
$( \exp( -(2r/\sigma^2)(X_t - X_0) - rt))_{t\in[0,T]},$
it is easy to see that
\[\E_x\!\left(e^{-r\tau}g(X_\tau) \right)
=e^{-(2r/\sigma^2)x}  \E_x\bigl( e^{(2r/\sigma^2) \widetilde X_\tau} g(\widetilde X_\tau) \bigr)
=e^{-(2r/\sigma^2)x}  \E_x\bigl(\widetilde g(\widetilde X_\tau) \bigr) \]
for any stopping time $\tau\leq T$.
Likewise, we have
$\E_x(e^{-r\theta}f(X_\theta))
=e^{-(2r/\sigma^2)x}\E_x( \widetilde f(\widetilde X_\theta)) $
for any European payoff function $f:\rr\to\rr_+$
and any $\theta\in[0,T]$ where $\smash{\widetilde f(x)}:=e^{(2r/\sigma^2)x}f(x).$
Some simple algebraic manipulations yield that
\begin{align}\label{EQ_TRANSFORM_BEIBEL_LERCHE}
 v_{\eu, \mu}(\theta, x)
&=e^{-r\theta}\int   \frac{ \nno\!\left(x +\widehat{r} \theta, \sigma^2 \theta, y \right)}{ \nno\!\left(x_0 +\widehat{r} T, \sigma^2 T, y \right)}    \mu(dy)\\
&= e^{(-2r/\sigma^2)x}\int  \frac{ \nno\!\left(x +\widetilde{r} \theta, \sigma^2 \theta, y \right)}{ \nno\!\left(x_0 +\widetilde{r} T, \sigma^2 T, y \right)}\nonumber  
e^{(2r / \sigma^2)x_0 - rT}   \mu(dy).
\end{align}
Consequently, the linear programme \eqref{PROGRAM_primal_general} is, up to renormalising the target functional, equivalent to 
\begin{equation}
\label{PROGRAM_P_0}
\begin{aligned}
& \text{minimise}   & & \Vert \mu \Vert  \\
& \text{subject to} & &   \mu \in \M^+(-\infty,K],\\
&&& \widetilde v_{\eu, \mu}(\theta,x)\geq \widetilde g(x)\text{ for any } (\theta,x)\in [0,T]\times(-\infty,K),
\end{aligned}
\end{equation}
where
\begin{align*}
 & \widetilde v_{\eu, \mu}(\theta, x) := \int  \frac{ \nno\!\left(x +\widetilde{r} \theta, \sigma^2 \theta, y \right)}{ \nno\!\left(x_0 +\widetilde{r} T, \sigma^2 T, y \right)}   \mu(dy), 
 & (\theta,x)\in (0,\infty)\times\rr, \\
 & \widetilde v_{\eu, \mu}(0, x) 	:=\liminf_{\substack{(\theta,y) \to (0,x)\\(\theta,y) \in (0,\infty)\times \rr} } \widetilde v_{\eu, \mu}(\theta, x), & x\in\rr.
\end{align*}
More specifically, any admissible measure $\mu$ in  \eqref{PROGRAM_primal_general}
corresponds to the admissible measure $e^{(2r / \sigma^2)x_0 - rT}\mu$ for the programme \eqref{PROGRAM_P_0}.
Note that $ \widetilde v_{\eu, \mu}(T, x_0) = \Vert \mu \Vert < \infty$ for any $\mu \in \M^+(\rr)$.

\subsubsection{Duality}\label{SubSection_DUALITY}
We define the set $\Omega := (0,T)\times (-\infty,K] $
and linear operators
\begin{align*}
 & \T: \M (-\infty,K] \to C( \Omega ) & \T\mu (t,x) := \int \kappa(t,x,y)    \mu(dy), \\
 & \T':  \M(\Omega) \to \B((-\infty,K],\rr) & \T' \lambda (y)  := \int \kappa(t,x,y)   \lambda(d(t,x))
\end{align*}
with the integral kernel
\begin{equation}\label{EQ_integrand_representation}
\begin{aligned}
 \kappa(t,x,y) &:=\frac{ \nno\!\left(x +\widetilde{r} (T-t), \sigma^2 (T-t), y \right)}{ \nno\!\left(x_0 +\widetilde{r} T, \sigma^2 T, y \right)}, \\
 &= \sqrt{\frac{T}{T-t}} \exp\!\left( - \frac{(y - A(x,t) )^2}{2 B(t)}    \right) \exp\!\left(  \frac{(x - x_0 - \widetilde{r} t)^2}{2 \sigma^2 t} \right)\\
 &=  \frac{T}{t} \frac{\nno(A(x,t),B(t),y)}{\nno(x_0 + \widetilde{r} t,\sigma^2 t,x)},
 \end{aligned}
\end{equation}
where $\B((-\infty,K],\rr)$ denotes the set of measurable functions from $(-\infty,K]$ to $\rr$.
and $A(t,x) := x_0 + (x-x_0)T/t$, $B(t):= \sigma^2 T (T-t) /t$, cf.\ Lemma \ref{LEMMA_norm_and_quotient_of_normal_pdf}(2).
Taking the specific structure of the integral kernel $\kappa$ into account, we can show that
for any measure $\mu\in\M (-\infty,K]$ the mapping $\Omega \ni (t,x)\mapsto \T\mu(t,x)$ is analytic on the open $\cc^2$-domain
\[ G:= \left\{ \theta \in \cc :  \sqrt{ (\Re \theta - T/2)^2 + (\Im \theta)^2 } <  T/2 \right\} \times \cc.\]
This is a special case of step 1 from Section \ref{section_european_representation} below, where a proof can be found.
In particular, the range of the operator $\T$ is indeed contained in $C( \Omega )$.

\begin{lemma}\label{LEMMA_injectivity_of_operator_T}
 If $\T\mu=0$ on some open subset of $\Omega$ then $\mu = 0$. In particular, the operator $\T$ is injective.
\end{lemma}
\begin{proof}
Let $\mu \in \M (-\infty,K]$ be a measure such that $\T\mu$ vanishes on some open subset of $\Omega$.
Denote by $\mu = \mu^+ - \mu^-$ the Hahn-Jordan decomposition of $\mu$. 
By analyticity and the identity theorem we conclude that $\T\mu = 0$ on $(0,T)\times \rr$.
Since $A(x + x_0/2, T/2) = 2x $ and $B(T/2)=\sigma^2 T$, we obtain from \eqref{EQ_integrand_representation} that
\[ \nno\!\left(\frac{x_0 +  \widetilde{r}T}{2} ,\sigma^2 \frac{T}{2} , x \right)  \T \mu^\pm \!\left(\frac{T}{2}, x + \frac{x_0}{2}\right)   
= 2 \int\nno\bigl(2x , \sigma^2 T  , y \bigr)   d\mu^\pm(y).\]
$\T\mu^+ = \T\mu^-$ implies
\[\int  \nno\bigl(y , \sigma^2 T  , x \bigr)   d\mu^+(y) = \int  \nno\bigl(y, \sigma^2 T , x \bigr)   d\mu^-(y) ,\quad x\in \rr.\]
Multiplying both sides with $e^{i zx}$ and integrating in $x$ yields
\[\int \exp\!\left( i yz - \frac{\sigma^2 T}{2} z^2 \right)    d\mu^+(y) = \int  \exp\!\left( i yz -  \frac{\sigma^2 T}{2} z^2 \right)  d\mu^-(y)\]
for all $z\in \rr$. Since the Fourier transform is injective, we conclude that the orthogonal measures $\mu^-$ and $\mu^+$ coincide.
This implies $\mu=0$ as desired.
\end{proof}

After these preliminary remarks we return to our optimisation problem. The convex programme \eqref{PROGRAM_P_0} can be rephrased in functional analytic terms as
\begin{equation*}
\tag{$P_0$} \label{PROGRAM_P_0_functional_analytic}
\begin{aligned}
& \text{minimise}   & & \Vert \mu \Vert \\
& \text{subject to} & &  \T\mu  - \widetilde{g} \in C^+(\Omega), \\
&                   & &  \mu \in \M^+(-\infty,K].
\end{aligned}
\end{equation*}
The requirement that the European value function dominates the payoff is expressed by the conic constraint.
To this primal minimisation problem we associate the Lagrange dual
\begin{equation}
\tag{$D_0$}\label{PROGRAM_D_0_functional_analytic}
\begin{aligned}
& \text{maximise}   & & \langle \widetilde{g}, \lambda \rangle\\
& \text{subject to} & & \T'\lambda(y) \leq 1  \quad\forall y \in (-\infty,K],\\
&                   & & \lambda \in \M^+(\Omega),
\end{aligned}
\end{equation}
where
$\langle \widetilde{g}, \lambda \rangle:=\int_\Omega \widetilde{g}(x) \lambda(d(t,x)).$

This dual problem allows for a probabilistic or physical interpretation.
To this end, suppose that particles move in space-time $\Omega\subset\rr_+\times\rr$, 
where the first coordinate of $(t,x)$ stands for time and the second for
the location at this time. In the space
coordinate $x$ the particles are assumed to follow a Brownian motion with drift rate $\widetilde{r}$ and diffusion coefficient $\sigma^2$.
Let us inject particles of total mass $\lambda(\Omega)$ into $\Omega$, distributed according to $\lambda$,
i.e.\ mass $\lambda(A)$ is assigned to any measurable subset $A$ of $\Omega$.
Where in  $\rr$ are the particles  to be found at the final time $T$?
Since they follow Brownian motion, they are distributed according to the 
Lebesgue density 
\[y\mapsto \int  \nno\bigl(x +  \widetilde{r} (T-t), \sigma^2  (T-t), y \bigr)  \lambda(d(t,x)).\]
On the other hand, the constraint 
$\int  \kappa(t,x,y)  \lambda(d(t,x)) \leq 1 $
can be rephrased as 
\begin{equation}\label{e:intuition}
\int   \nno\bigl(x + \widetilde{r}  (T-t), \sigma^2  (T-t), y \bigr) \lambda(d(t,x)) \leq 
\nno\bigl(x_0 +  \widetilde{r} T, \sigma^2 T, y \bigr).
\end{equation}
The right-hand side is the 
probability density function at time $T$ of a Brownian motion
started in $x_0$ at time $0$.
Put differently, the constraint (\ref{e:intuition}) means that
we consider only laws $\lambda$ on space-time $\Omega$
such that the resulting final distribution on
$\rr$ is dominated by the Gaussian law stemming from
a Brownian motion started in $x_0$ at time $0$.
If equality holds in (\ref{e:intuition}), the distribution of particles at time $T$ is the same as for a Brownian motion 
with drift rate $\widetilde{r}$, diffusion coefficient $\sigma^2$, and starting in $x_0$ at time 0. 

Regarding the primal problem \ref{PROGRAM_P_0_functional_analytic} and its formal dual \ref{PROGRAM_D_0_functional_analytic}, we may wonder whether weak or even strong duality holds, if
optimisers exist and if they are linked by some
complementary slackness condition. The following first main result shows that this is indeed the case, at least if 	
the CDEO payoff strictly dominates the American payoff function $\widetilde{g}$ at all $x<K$.

\begin{lemma}\label{Theorem_summary1} 
\begin{enumerate}
\item The optimal value of \ref{PROGRAM_P_0_functional_analytic} is obtained by some $\mu_0\in \M^+(-\infty,K]$ 
and it coincides with the optimal value of \ref{PROGRAM_D_0_functional_analytic}. 
The measure $\mu_0$ puts mass on every open subset of $(-\infty, K)$.
\item 
If $\widetilde{v}_{\eu, \mu_0}(0, x) > \widetilde g(x)$ for any $x \in (-\infty,K)$, 
the optimal value of \ref{PROGRAM_D_0_functional_analytic} is obtained by some measure $\lambda_0 \in \M^+(\Omega)$. 
In this case the following complementary slackness conditions are satisfied:
\begin{eqnarray}\label{e:EQ_CS_lambda}
\T\mu_0(\theta,x) &=& \widetilde{g}(x)  \quad \lambda_0\text{-a.e.\ on }\Omega,\\
\T'\lambda_0(x)&=&1   \quad \quad \mu_0\text{-a.e.\ on }(-\infty,K]\label{e:EQ_CS_mu}.
\end{eqnarray}
\end{enumerate}
\end{lemma}
In view of the discussion from Subsection \ref{SubSubSection_Transformation} this theorem can be easily
restated in terms of the quantities $v_{\eu,\mu}$ and $g$ associated to the  
programme \eqref{PROGRAM_primal_general}. We immediately obtain Theorem \ref{t:Theorem_summary0} from the first assertion of Lemma \ref{Theorem_summary1}.

\subsubsection{Proof of Lemma \ref{Theorem_summary1}}\label{SUBSECTION_Proof_Theorem_summary1}
For any $\epsilon \in (0,T)$ we define the set
\[ \Omega_\epsilon := [\epsilon,T-\epsilon] \times [-1/\epsilon, K]\]
and the following linear operator:
\begin{align*}
  \T^*:  \M(\Omega_\epsilon) \to C_0(-\infty,K] &\qquad\qquad\qquad \T^* \lambda (y)  := \int \kappa(t,x,y)   \lambda(d(t,x))
\end{align*}
The range of the operator $\T^*$ is contained in $C_0(-\infty,K]$ due to  
Lebesgue's dominated convergence theorem, by \eqref{EQ_integrand_representation} and
the compactness of the set $\Omega_\epsilon$. 
On the Cartesian products $C(\Omega_\epsilon) \times \M(\Omega_\epsilon)$ and 
$C_0(-\infty,K] \times \M (-\infty,K]$ we consider the algebraic pairing
\begin{equation}
 \label{EQ_Pairing_on_cartesian_products}
  \langle f , \nu \rangle \mapsto \int f d\nu.
\end{equation}
This mapping is finitely valued, bilinear and separates points.
We endow $C(\Omega_\epsilon), \M(\Omega_\epsilon)$ and $\M (-\infty,K]$ with the weak topologies $\sigma(C,\M), \sigma(\M,C)$ and $\sigma(\M,C_0)$ 
induced by \eqref{EQ_Pairing_on_cartesian_products}.
The function space $C_0 (-\infty,K]$ is endowed with the topology of uniform convergence $\mathscr{T}_{\mathrm{uc}}$. 
This turns all four spaces into locally convex Hausdorff spaces. 
Moreover, each space of measures is the continuous dual of the associated function space and vice versa, cf.\ \cite[Theorem 6.19]{rudin.87}.
Fubini's theorem yields that for all measures $\mu \in \M (-\infty,K]$ and $\lambda \in \M(\Omega_\epsilon)$ we have
\begin{equation}\label{EQ_Adjoint_equation_of_operator_T}
 \langle \T\mu , \lambda \rangle = \langle \mu, \T^* \lambda \rangle.
\end{equation}
By \cite[Lemma 5.17]{lenga.17} we find that
the operator $\T$ is $\sigma(\M,C_0)$-$\sigma(C,\M)$ continuous and $\T^*$ is $\sigma(\M,C)$-$\sigma(C_0,\M)$ continuous.

We want to find a measure $\mu_0 \in \M^+(-\infty,K]$ which solves the linear programme \ref{PROGRAM_P_0_functional_analytic} from Subsection \ref{SubSection_DUALITY}.
Our strategy is to approximate this optimisation problem by the following sequence of linear programmes with milder constraints
\begin{equation}
\tag{$P_\epsilon$} \label{PROGRAM_P_eps}
\begin{aligned}
& \text{minimise}   & & \Vert \mu \Vert \\
& \text{subject to} & & (\left.\T\mu - \widetilde{g}\right)|_{\Omega_\epsilon} \in C^+(\Omega_\epsilon), \\
&                   & &  \mu \in \M^+(-\infty,K].
\end{aligned}
\end{equation}
The solution to \ref{PROGRAM_P_0_functional_analytic} will be obtained by compactness from the family of \ref{PROGRAM_P_eps}-extremal elements. 
For each $\epsilon\in (0,T/2)$, the Lagrange dual problem of \ref{PROGRAM_P_eps} is given by
\begin{equation}
\tag{$D_\epsilon$}\label{PROGRAM_D_eps}
\begin{aligned} 
& \text{maximise}   & & \langle \widetilde{g}, \lambda \rangle\\
& \text{subject to} & & 1 -\T^*\lambda \in C^+(-\infty,K],\\
&                   & & \lambda \in \M^+(\Omega_\epsilon).
\end{aligned}
\end{equation}
The optimal values of \ref{PROGRAM_P_eps} and \ref{PROGRAM_D_eps} are denoted by $p_\epsilon$ and $d_\epsilon$, respectively.
By construction we have that \emph{weak duality} $0 \leq d_\epsilon \leq p_\epsilon$  holds. Indeed, in view of the adjointness relation
\eqref{EQ_Adjoint_equation_of_operator_T} we obtain
\begin{equation}\label{EQ_Weak_duality_Deps_Peps}
 0 \leq \langle \widetilde{g}, \lambda \rangle \leq  \langle \T\mu, \lambda \rangle = \langle \mu, \T^*\lambda \rangle 
\leq \langle \mu, 1 \rangle = \Vert \mu\Vert
\end{equation}
for any primal admissible $\mu \in \M^+(-\infty,K]$ and any dual admissible $\lambda \in \M^+(\Omega_\epsilon)$.
Next, we verify \emph{primal and dual attainment}.
The nonnegative measure $\widetilde{\mu}$ with Lebesgue density
\begin{align*}
y \mapsto   2 \Vert \widetilde{g} \Vert_\infty   \nno\!\left(x_0 +\widetilde{r} T, \sigma^2 T, y \right) 1_{(-\infty,K)}(y)
\end{align*}
is \ref{PROGRAM_P_eps}-admissible because for any $(t,x) \in  \Omega_\epsilon$ we have
\begin{equation}\label{EQ_Masse_mu_epsilon_abschaetzen}
\begin{aligned}
\T  \widetilde{\mu} (t,x)
&= 2 \Vert \widetilde{g} \Vert_\infty  \int_{-\infty}^K \nno\bigl(x +\widetilde{r} (T-t), \sigma^2 (T-t), y \bigr)dy \\ 
&= 2 \Vert \widetilde{g} \Vert_\infty   \Phi\!\left( \frac{K-x}{\sigma\sqrt{T-t}} - \widetilde{r} \frac{\sqrt{T-t}}{\sigma} \right) \\ 
&\geq 2 \Vert \widetilde{g} \Vert_\infty   \Phi\!\left( 0 \right) = \Vert \widetilde{g} \Vert_\infty.
\end{aligned}
\end{equation}
The total mass of the measure $\widetilde{\mu}$ is bounded by the constant $2\Vert \widetilde{g} \Vert_\infty $. 
Therefore solving the minimisation problem \ref{PROGRAM_P_eps} is equivalent to minimising the total variation norm over the $\sigma(\M,C_0)$-compact set
\begin{equation}\label{EQ_SET_C_eps_p}
C^\epsilon_p  :=   \T^{-1}\!\left(  \widetilde{g} + C^+(\Omega_\epsilon)  \right)  \cap   \M^+(-\infty,K]  \cap  B_{\M(\rr)}(0,2\Vert \widetilde{g}\Vert_\infty).
\end{equation}
The $\sigma(\M,C_0)$-compactness of $C^\epsilon_p$ is established as follows.
First we note that the set $\widetilde{g} + C^+(\Omega_\epsilon)$ is homeomorphic to the 
$\sigma(C,\M)$-closed cone 
\[C^+(\Omega_\epsilon) = \bigcap_{\lambda \in \M^+(\Omega_\epsilon)} \{ f \in C(\Omega_\epsilon): \langle f, \lambda \rangle \geq 0  \}\] 
and that the continuity properties of the operator $\T$ warrant the $\sigma(\M,C_0)$-closedness of the preimage $\T^{-1}\!\left( \widetilde{g} + C^+(\Omega_\epsilon)\right)$.
Secondly, we observe that the cone
\[\M^+(-\infty,K] =  \bigcap_{f \in C_0^+(-\infty,K]} \{ \mu \in \M(-\infty,K] : \langle f, \mu \rangle \geq 0\}\]
is $\sigma(\M,C_0)$-closed as well and that $B_{\M(\rr)}(0,2\Vert \widetilde{g}\Vert_\infty)$ is a $\sigma(\M,C_0)$-compact set due to Alaoglu's theorem, cf.\ \cite[Theorem 23.5]{meise.vogt.97}.
The target functional  $\mu \mapsto \Vert \mu \Vert$ is lower semi-continuous with respect
to the topology $\sigma(\M,C_0)$ and therefore its minimal value $p_\epsilon$ is attained
by some measure $\mu_\epsilon \in C^\epsilon_p$.

Next, we prove the attainment of the \ref{PROGRAM_D_eps}-optimal value.
For any measure $\lambda \in \M(\Omega_\epsilon)$ and $y \in (-\infty,K]$ we define 
$\U \lambda (y) := \nno\!\left(x_0 + \widetilde{r} T, \sigma^2 T, y \right) \T^*\lambda(y)$.
Obviously $\U$ is a $\sigma(\M,C)$-$\sigma(C_0,\M)$-continuous, linear operator from $\M(\Omega_\epsilon)$ into the space $C_0 (-\infty,K]$.
The inequality constraint of the programme\ref{PROGRAM_D_eps} is equivalent to
$\U\lambda(y) \leq   \nno\bigl(x_0 + \widetilde{r} T, \sigma^2 T, y \bigr)$
for all $y \in (-\infty,K]$. Integrating this inequality over the interval $ (-\infty,K]$ yields
\[    \int \int_{-\infty}^K   \nno\bigl(x + \widetilde{r} (T-t), \sigma^2 (T-t), y \bigr)    dy   \lambda(d(t,x)) 
\leq  \int_{-\infty}^K  \nno\bigl(x_0 + \widetilde{r} T, \sigma^2 T, y \bigr)   dy \leq 1.\]
A calculation similar to \eqref{EQ_Masse_mu_epsilon_abschaetzen} yields
\[  \int_{-\infty}^K  \nno\bigl(x + \widetilde{r} (T-t), \sigma^2 (T-t), y \bigr)  dy 
\geq \Phi\!\left(   0 \right) = \frac{1}{2} \]
for any $(t,x)\in \Omega_\epsilon$ and consequently any \ref{PROGRAM_D_eps}-admissible measure $\lambda$  satisfies $\Vert \lambda \Vert \leq 2$.
Solving the maximisation problem \ref{PROGRAM_D_eps} is therefore equivalent to maximising the
$\sigma(\M,C)$-continuous mapping $\lambda \mapsto \langle \widetilde{g} , \lambda \rangle$ over the set
\[C^\epsilon_d  :=   \U^{-1}\!\left(   \nno\bigl(x_0 + \widetilde{r} T, \sigma^2 T, \cdot \bigr)  - C_0^+(-\infty,K]  \right)  \cap  \M^+(\Omega_\epsilon)  
\cap  B_{\M(\Omega_\epsilon)}(0,2). \]
One easily modifies the arguments following \eqref{EQ_SET_C_eps_p} in order to 
verify that $C^\epsilon_d$ is a $\sigma(\M,C)$-compact subset of $\M(\Omega_\epsilon)$.
Hence the target functional of the Lagrange dual \ref{PROGRAM_D_eps} attains its maximal value $d_\epsilon$ at some measure $\lambda_\epsilon \in C^\epsilon_d$.

In order to prove \emph{strong duality} $d_\epsilon = p_\epsilon$, we use some well-established techniques from convex optimisation. 
We refer the reader to \cite{rockafellar.74} for a well-written introduction to conjugate duality and optimisation on paired spaces. 
A short summary for our needs can be found in \cite[Section 5.4]{lenga.17}.
The Lagrange function $K :  \M(-\infty,K] \times \M(\Omega_\epsilon) \to [-\infty, \infty]$ associated to the \ref{PROGRAM_P_eps}-\ref{PROGRAM_D_eps}-duality is defined by
\begin{equation}\label{EQ_Lagrange_function}
  K(\mu, \lambda ) :=  \Vert \mu \Vert + \langle \widetilde{g}, \lambda \rangle - \langle \T \mu , \lambda \rangle + \I_{\M^+(-\infty,K]}(\mu) - \I_{\M^+(\Omega_\epsilon)}(\lambda),
\end{equation}
where \[\I_M(x) :=\begin{cases} 0 &\text{if } x \in M, \\ \infty &\text{if } x \notin M \end{cases}\] for any set $M$.
For later reference we provide the following explicit calculations:
\begin{align}
\sup_{\lambda \in \M(\Omega_\epsilon)} \inf_{\mu \in \M(-\infty,K]}  K( \mu, \lambda )  
&= \sup_{\lambda \in \M^+(\Omega_\epsilon)} \inf_{\mu \in \M^+(-\infty,K]}   \left(  \Vert\mu\Vert  + \langle \widetilde{g} - \T\mu   , \lambda \rangle  \right) \label{EQ_MINIMAX_1} \\
&= \sup_{\lambda \in \M^+(\Omega_\epsilon)} \left( \langle \widetilde{g}  , \lambda \rangle +  \inf_{\mu \in \M^+(-\infty,K]}   \langle  1 - \T^* \lambda    , \mu \rangle  \right) \nonumber \\
&= \sup_{\substack{\lambda \in \M^+(\Omega_\epsilon) \\  \T^* \lambda \leq 1 }} \langle \widetilde{g}, \lambda \rangle = d_\epsilon,  \\\nonumber\\
\inf_{\mu \in \M(-\infty,K]} \sup_{\lambda \in \M(\Omega_\epsilon)}  K( \mu, \lambda )  
&=  \inf_{\mu \in \M^+(-\infty,K]} \sup_{\lambda \in \M^+(\Omega_\epsilon)}  \left(  \Vert\mu\Vert  + \langle \widetilde{g} - \T\mu   , \lambda \rangle  \right) \label{EQ_MINIMAX_2}  \\
&=  \inf_{\mu \in \M^+(-\infty,K]} \left(  \Vert\mu\Vert   +  \sup_{\lambda \in \M^+(\Omega_\epsilon)}   \langle  \widetilde{g} - \T\mu, \lambda\rangle  \right) \nonumber  \\
&= \inf_{\substack{ \mu \in \M^+(-\infty,K] \\  \T\mu \geq \widetilde{g} } }  \Vert \mu \Vert =  p_\epsilon. \nonumber 
\end{align}
One easily verifies that the mapping  $\M(-\infty,K] \ni \mu \mapsto K_\lambda(\mu) := K(\mu,\lambda)$ 
is closed in the sense of \cite[Section 3]{rockafellar.74} and convex for any $\lambda \in \M(\Omega_\epsilon)$.
\begin{lemma}\label{LEMMA_dual_value_function_properties}
The \emph{dual value function} $v: C_0(-\infty,K] \mapsto (-\infty,\infty]$,
\[v(f) := \inf_{\lambda \in \M(\Omega_\epsilon)} K_\lambda^*(f) \]
is convex and we have $v(0) = -d_\epsilon \geq v^{**}(0)= -p_\epsilon$. 
Here $K_\lambda^*$ denotes the convex conjugate of the mapping $K_\lambda$. 
\end{lemma}
\begin{proof}
By  Fenchel's inequality  and \eqref{EQ_MINIMAX_1} we have
\[v^{**}(0) \leq v(0) = \inf_{\lambda \in \M(\Omega_\epsilon)} K_\lambda^*(0) 
= -\sup_{\lambda \in \M(\Omega_\epsilon)} \inf_{\mu \in \M(-\infty,K]}  K( \mu, \lambda )   
= -d_\epsilon.\]
The conjugate $v^* : \M(-\infty,K] \mapsto [-\infty, \infty]$ of the function $v$ is given by
\begin{align*}
v^*(\mu) &:= \sup_{f \in C_0(-\infty,K]} \left( \langle f , \mu   \rangle  - v(f)  \right)\\
&= \sup_{\lambda \in \M(\Omega_\epsilon)} \sup_{f \in C_0(-\infty,K]} \left( \langle f , \mu   \rangle  -  K_\lambda^*(f) \right)\\
&= \sup_{\lambda \in \M(\Omega_\epsilon)}  K_\lambda^{**}(\mu)\\
&= \sup_{\lambda \in \M(\Omega_\epsilon)}  K(\mu, \lambda).
\end{align*}
The last equality follows from the Fenchel-Moreau theorem because the mapping $ K_\lambda$ is closed and convex, cf.\ \cite[Theorem 5]{rockafellar.74}.
Hence the biconjugate of the dual value function is given by 
\begin{align}
v^{**}(f) &:=  \sup_{\mu \in \M(-\infty,K]} \left( \langle f , \mu   \rangle  -v^*(\mu)  \right) \nonumber\\
&=  \sup_{\mu \in \M(-\infty,K]}\inf_{\lambda \in \M(\Omega_\epsilon)} \left( \langle f , \mu   \rangle  -  K(\mu, \lambda)  \right).\label{EQ_biconjugate_of_dual_value_function}
\end{align}
Owing to \eqref{EQ_MINIMAX_2} we obtain $v^{**}(0)= -p_\epsilon$. 

Next, we show that the mapping $v$ does not assume the value $-\infty$. 
Suppose to the contrary that there exists some $f \in  C_0(-\infty,K]$ with $v(f) = -\infty$. 
Fenchel's inequality yields $v^{**} \leq v$
and hence $v^{**}(f) = -\infty$. Equation \eqref{EQ_biconjugate_of_dual_value_function} now implies that
$\sup_{\lambda \in \M(\Omega_\epsilon)}  K(\mu,\lambda) = \infty$
for any measure $\mu \in \M(-\infty,K]$ and therefore $p_\epsilon = \infty$. 
This is impossible because the set of \ref{PROGRAM_P_eps}-admissible measures has already been shown to be nonempty.

In order to verify that $v$ is convex, suppose that $\alpha\in (0,1)$ and $f_1,f_2\in C_0(-\infty,K]$.
From \eqref{EQ_Lagrange_function} it is apparent that the Lagrange function $K$ 
is concave in the second component and this yields
\begin{equation}\label{EQ_convexity_of_dual_value_function}
\begin{aligned}
v( \alpha f_1 + (1- \alpha) f_2 ) 
&=\inf_{\lambda \in \M(\Omega_\epsilon)} \sup_{\mu \in \M(-\infty,K]} \bigl(  \langle \alpha f_1 + (1- \alpha) f_2   , \mu \rangle - K( \mu, \lambda )  \bigr)\\
&\leq \sup_{\mu \in \M(-\infty,K]} \bigl(  \langle \alpha f_1 + (1- \alpha) f_2  , \mu \rangle - K( \mu, \alpha\lambda_1 + (1-\alpha)\lambda_2 )  \bigr)\\
&\leq \alpha \sup_{\mu \in \M(-\infty,K]} \bigl(  \langle  f_1  , \mu \rangle - K( \mu, \lambda_1 )  \bigr) + (1-\alpha) \sup_{\mu \in \M(-\infty,K]} \bigl(  \langle  f_2 , \mu \rangle - K( \mu, \lambda_2 )  \bigr) 
\end{aligned}
\end{equation}
for any choice of $\lambda_0,\lambda_1 \in \M(\Omega_\epsilon)$. 
Minimising with respect to $\lambda_0, \lambda_1$ proves that $v$ is indeed a convex function. 
\end{proof}
Hence strong duality holds if we can show that $v^{**}(0)=v(0)$ is true. 
By virtue of Lemma \ref{LEMMA_dual_value_function_properties} and
the Fenchel-Moreau biconjugate theorem, cf.\ \cite[Theorem 5]{rockafellar.74}, we obtain
\[v^{**} (0) =  \lsc(v)(0) = \sup_{O \in \mathscr{U}(0)} \inf_{f \in O \setminus \{0\}} v(f),\]
where $\lsc(v)$ denotes the semi-continuous hull of the mapping $v$, cf.\ \cite[Equation 3.7]{rockafellar.74}
and $\mathscr{U}(0)$ the set containing all $\mathscr{T}_{\mathrm{uc}}$-open neighbourhoods of $0$. 
Put differently, in order to verify  strong duality it suffices to 
show that the mapping $v$ is continuous at the origin with respect to the topology of uniform convergence.
We  use the following adaptation of \cite[Theorem 5.42]{aliprantis.border.06} to locally convex spaces:
\begin{lemma}\label{lemma_stetigkeit_konvexer_funktionen} 
Let $V$ be a locally convex space,  $f: V \to (-\infty, \infty]$ a convex function and $x_0 \in V$.
If there exists an open neighbourhood $O$ of $x_0$ such that $\sup_{x\in O} f(x) < \infty$, then $f$ is continuous at $x_0$.
\end{lemma}
The set $O:=\left\{ \Vert f \Vert_\infty < 1 \right\} $ is a $\mathscr{T}_{\mathrm{uc}}$-open neighbourhood of $0$.
For any $f \in O$ we have
\begin{align*}
v(f) &= \inf_{\lambda \in \M^+(\Omega_\epsilon)} \sup_{\mu \in \M^+(-\infty,K]}
\bigl(  \langle f  , \mu \rangle - \Vert \mu \Vert - \langle \widetilde{g}, \lambda \rangle + \langle \T \mu , \lambda \rangle  \bigr)\\
&\leq \sup_{\mu \in \M^+(-\infty,K]} \bigl( \Vert\mu\Vert \Vert f \Vert_\infty   -  \Vert\mu\Vert  \bigr) = 0.
\end{align*}
Lemma \ref{lemma_stetigkeit_konvexer_funktionen} warrants that the mapping $v$ is indeed continuous at $0$ and therefore \[p_\epsilon=-v^{**}(0)=-v(0)= d_\epsilon.\]

Next, we verify that the optimisers $\lambda_\epsilon$ and $\mu_\epsilon$ satisfy the \emph{complementary slackness} property. 
Using the strong duality we obtain
\begin{equation}\label{COMPLEMENTARY_SLACKNESS_Peps_Deps}
	0 	\leq \langle \T \mu_\epsilon - \widetilde{g} , \lambda_\epsilon \rangle
	= \langle \T \mu_\epsilon , \lambda_\epsilon \rangle - p_\epsilon
	= \langle  \mu_\epsilon , \T^* \lambda_\epsilon \rangle - d_\epsilon
	= \langle \mu_\epsilon,  \T^* \lambda_\epsilon  - 1 ,  \rangle 
	\leq  0.
\end{equation}
In other words, $\T \mu_\epsilon = \widetilde{g}$ holds $\lambda_\epsilon$-a.e.\ on $\Omega_\epsilon$ and 
$\T^* \lambda_\epsilon  = 1$ holds $\mu_\epsilon$-a.e.\ on $(-\infty,K]$. 
Moreover, the structure of the dual problem \ref{PROGRAM_D_eps} implies that we can always choose a
\ref{PROGRAM_D_eps}-optimal element which assigns no mass to the zeros of the function $\widetilde{g}$,
i.e.\ $\lambda_\epsilon(\{(t,x) \in \Omega_\epsilon : \widetilde{g}(x) = 0\} ) = 0$.
From now on we will only consider dual maximisers with this property. 

Let us summarise the findings from above:
\begin{lemma}
For any $\epsilon \in (0,T/2)$ the linear programmes \ref{PROGRAM_P_eps}, \ref{PROGRAM_D_eps} have solutions $\mu_\epsilon, \lambda_\epsilon$ and  
their optimal values $p_\epsilon, d_\epsilon$ coincide.  The total mass of both optimisers is bounded by a constant $\rho \in (0,\infty) $ that does not depend on $\epsilon$.
Moreover, no mass of the measure $\lambda_\epsilon$ is located on the zero set of the function $\widetilde{g}$. The equation $\T \mu_\epsilon = \widetilde{g}$ holds $\lambda_\epsilon$-a.e.\ on $\Omega_\epsilon$ and 
$\T^* \lambda_\epsilon  = 1$ holds $\mu_\epsilon$-a.e.\ on $(-\infty,K]$.
\end{lemma}

We now turn our attention to programme \ref{PROGRAM_P_0_functional_analytic} and 
the associated dual \ref{PROGRAM_D_0_functional_analytic} from Subsection \ref{SubSection_DUALITY}.
Lemma \ref{Theorem_summary1} is proved in two steps. First, we show that the primal
optimisers $(\mu_\epsilon)_{\epsilon>0}$ cluster at some \ref{PROGRAM_P_0_functional_analytic}-optimal measure $\mu_\epsilon$
and that the family $(\lambda_\epsilon)_{\epsilon>0}$ contains a \ref{PROGRAM_D_0_functional_analytic}-admissible accumulation point $\lambda_0$.
Subsequently we show that the measure $\lambda_0$ is \ref{PROGRAM_D_0_functional_analytic}-optimal.
The other assertions of Lemma \ref{Theorem_summary1} are verified on the way.

\paragraph{Step 1} Let $p_0$ and $d_0$ denote the optimal values of \ref{PROGRAM_P_0_functional_analytic} and \ref{PROGRAM_D_0_functional_analytic}.
The weak duality $0 \leq d_0 \leq p_0$ follows literally from the same calculation as in \eqref{EQ_Weak_duality_Deps_Peps}.
Recall that for any $\epsilon > 0$ the mass of the optimisers $\mu_\epsilon \in \M^+(-\infty,K]$ and  $\lambda_\epsilon \in \M^+(\Omega)$ 
is bounded by some constant $\rho>0$, which does not depend on $\epsilon$.
General theory tells us that the vague topology is metrisable on the total variation unit balls in both spaces.
Alaoglu's theorem warrants that they are vaguely compact sets.	
Hence we can find a sequence $ \epsilon_n \downarrow 0$ and measures  
$\mu_0 \in \M^+(-\infty,K], \lambda_0 \in \M^+(\Omega)$ with $\Vert \mu_0 \Vert \vee \Vert \lambda_0 \Vert \leq \rho$
such that $\mu_{\epsilon_n}$ converges vaguely to $\mu_0$ and $\lambda_{\epsilon_n}$ converges vaguely to $\lambda_0$.
For any  $(t,x) \in \Omega$ the mapping $y \mapsto   \kappa(t,x,y) $ is continuous on $(-\infty,K]$ and vanishes at infinity, 
see \eqref{EQ_integrand_representation}.
By vague convergence we conclude that
\begin{align*}
\T\mu_0(t,x)  
=  \int   \kappa(t,x,y)  d\mu_0(y) 
=  \lim_{n \to \infty}  \int   \kappa(t,x,y)  d\mu_{\epsilon_n}(y)
\geq \widetilde{g}(x), \quad (t,x) \in \Omega.
\end{align*}
This ensures that $\mu_0$ is indeed \ref{PROGRAM_P_0_functional_analytic}-admissible.  
Next we verify that the measure $\lambda_0$ is \ref{PROGRAM_D_0_functional_analytic}-admissible.
Obviously, for any $\delta\in (0,T/4)$ we have $\emptyset \not= \Omega_{2\delta} \subset \Omega_{\delta} \subset \Omega$.
By Urysohn's lemma, cf.\ \cite[Theorem 4.2]{lang.93}, there exists
a continuous function $\phi^\delta : \Omega \to [0,1]$ such that $\phi^\delta(t,x) = 1$ for
all $(t,x) \in \Omega_{2\delta}$ and $\phi^\delta(t,x) = 0$ for all $(t,x) \in \cl(\Omega \setminus \Omega_{\delta})$.
For any $y \in (-\infty,K]$ the continuous mapping 
$ \Omega \ni (x,t) \mapsto \kappa(t,x,y)  \phi^{\delta}(t,x)$
vanishes at infinity. By vague convergence of the sequence $\lambda_{\epsilon_n} \to \lambda_0$ we obtain 
\begin{align*}
\int  \kappa(t,x,y)  \lambda_0(d(t,x))
&=    \lim_{\delta \downarrow 0} \int \kappa(t,x,y) 1_{\Omega_{2\delta}}(t,x)  \lambda_0(d(t,x))\\
&\leq \lim_{\delta \downarrow 0} \int \kappa(t,x,y) \phi^{\delta}(t,x)  \lambda_0(d(t,x))\\
&=    \lim_{\delta \downarrow 0} \lim_{n \to \infty} \int  \kappa(t,x,y) \phi^{\delta}(t,x)  \lambda_{\epsilon_n}(d(t,x))\\
&\leq  \limsup_{\delta \downarrow 0} \limsup_{n \to \infty} \int \kappa(t,x,y)  \lambda_{\epsilon_n}(d(t,x)) \leq 1.
\end{align*}
In other words, the measure $\lambda_0$ is dual admissible.

Next, we establish the strong duality $p_0 = d_0$ by putting together several of the previous results. 
The vague convergence of the measures $\mu_{\epsilon_n}$ to $\mu_0$ implies that $\Vert \mu_0\Vert \leq \liminf_{n\to\infty} \Vert \mu_{\epsilon_n}\Vert$ is true.
Recalling that strong duality holds in the \ref{PROGRAM_P_eps}-\ref{PROGRAM_D_eps}-setting yields
\begin{equation}\label{EQ_strong_duality_p0_d0}
d_0 \leq p_0 \leq \Vert \mu_0\Vert \leq \liminf_{n\to\infty} \Vert \mu_{\epsilon_n}\Vert 
= \liminf_{n\to\infty} p_{\epsilon_n} = \liminf_{n\to\infty} d_{\epsilon_n} \leq d_0.
\end{equation}
The last inequality follows from the fact that all \ref{PROGRAM_D_eps}-admissible elements are \ref{PROGRAM_D_0_functional_analytic}-admissible.
Along the way we have shown that the \ref{PROGRAM_P_0_functional_analytic}-optimal value is attained by the measure $\mu_0$. 

We prove by contradiction that any \ref{PROGRAM_P_0_functional_analytic}-admissible element assigns mass to any open subset of $(-\infty,K)$.
Otherwise there is some \ref{PROGRAM_P_0_functional_analytic}-admissible measure $\mu$ 
and a bounded, open interval $I:= (c - \nu, c + \nu) \subset (-\infty,K)$ such that $\mu(I)=0$. 
Obviously we have  $0< \delta := \inf_{x\in I} \widetilde{g}(x)$. This yields
\begin{equation}\label{EQ_always_put_mass_on_open_sets}
\delta < \widetilde{g}(c) 
\leq \T\mu(t,c)
= \int    1_{\{ \vert y-c \vert \geq \nu \}} \kappa(t,c,y)   \mu(dy)
\end{equation}
for all $t\in (0,T)$. 
In view of \eqref{EQ_integrand_representation}, the right-hand side of \eqref{EQ_always_put_mass_on_open_sets}
converges to $0$ as $t \uparrow T$. This contradiction proves the claim.

\paragraph{Step 2} We show that the \ref{PROGRAM_D_0_functional_analytic}-optimal value  is attained by $\lambda_0$ if some additional requirement is met. 
Recall that the measure $\lambda_0$ is \ref{PROGRAM_D_0_functional_analytic}-admissible and that the sequence $\lambda_{\epsilon_n}$ 
converges to $\lambda_0$ with respect to the vague topology on $\M(\Omega)$. Due to the lack of compactness, we cannot
directly conclude that $\lambda_{\epsilon_n}$ converges weakly to $\lambda_0$.
Observe that the functional $ \M(\Omega) \ni \lambda  \mapsto \langle \widetilde{g}, \lambda \rangle$ is weakly but not vaguely continuous. 

First we prove that the sequence $\lambda_{\epsilon_n}$ converges weakly in $\M(\cl\,\Omega)$,
where $\cl\,\Omega = [0,T] \times (-\infty,K]$. It is sufficient to show
that the family $\{\lambda_{\epsilon_n} : n \in \nn \}$ is tight. 
For any $\epsilon > 0$ we define by $K_\epsilon := [0,T] \times [-1/\epsilon,K]$ a compact subset of $\cl\,\Omega$. 
The mass of $\lambda_{\epsilon_n}$ is concentrated on $\Omega_{\epsilon_n} \subset K_{\epsilon_n}$. 
Let us assume by contradiction that the family of measures is not tight. 
Then there exists a constant $\delta >0 $ such that for any $n \in \nn$ there is some integer $M_n \geq n$ with
$\lambda_{\epsilon_{M_n}} ( \Omega \setminus K_{\epsilon_n}) > \delta$.
Pick a sufficiently small constant $C\in (-\infty,K)$ with
\[ \int_{-\infty}^C   \nno\bigl(x_0 +\widetilde{r} T, \sigma^2 T, y \bigr)   dy  \leq \frac{\delta}{2}.\]
Due to the fact that all measures $\lambda_{\epsilon_n}$ are \ref{PROGRAM_D_0_functional_analytic}-admissible, we have
\[ \int  \nno\bigl(x +\widetilde{r} (T-t), \sigma^2 (T-t), y \bigr)    \lambda_{\epsilon_{M_n}}(d(t,x)) \leq  \nno\bigl(x_0 +\widetilde{r} T, \sigma^2 T, y \bigr) \]
for any $y\in (-\infty,K]$. Integrating this inequality over the set $(-\infty,C)$ yields
\[ \int \int_{-\infty}^C    \nno\bigl(x +\widetilde{r} (T-t), \sigma^2 (T-t), y \bigr)  dy  \lambda_{\epsilon_{M_n}}(d(t,x)) \leq  \frac{\delta}{2}.\]
Due to the positivity of measure and integrand, we conclude that
\begin{align*}
\frac{\delta}{2} 
&\geq \int_{\Omega \setminus K_{\epsilon_{n} }}     \int_{-\infty}^C    \nno\bigl(x +\widetilde{r} (T-t), \sigma^2 (T-t), y \bigr) dy   \lambda_{\epsilon_{M_n}}(d(t,x))\\
&\geq  \lambda_{\epsilon_{M_n}}\left(\Omega \setminus K_{\epsilon_{n}} \right)  
\inf_{(t,x) \in \Omega \setminus K_{\epsilon_{n} }} \int_{-\infty}^C   \nno\bigl(x +\widetilde{r} (T-t), \sigma^2 (T-t), y \bigr) dy \\
&\geq  \delta 
\inf_{ x  < -1/\epsilon_n } \inf_{t \in [0,T]} \int_{-\infty}^C    \nno\bigl(x +\widetilde{r} t, \sigma^2 t, y \bigr) dy, \quad n\in\nn.
\end{align*}
Taking the limit $n \to \infty$ yields
\[ \frac{\delta}{2}  \geq \delta \lim_{n \to \infty}  \inf_{ x  < -1/\epsilon_n } \inf_{t \in [0,T]}  \int_{-\infty}^C   \nno\bigl(x +\widetilde{r} t, \sigma^2 t, y \bigr)  dy  = \delta\]
as $\epsilon_{n} \to 0$.
This is impossible and consequently the family $\{\lambda_{\epsilon_n} : n \in \nn \}$ must be tight. 
Hence  the sequence $\lambda_{\epsilon_n}$ converges weakly in $\M(\cl\,\Omega)$ to some measure $\overline{\lambda}_0$ with $\overline{\lambda}_0 \vert_\Omega = \lambda_0$.

In order to assure that the measure $\lambda_0$ is \ref{PROGRAM_D_0_functional_analytic}-optimal, it suffices to show that $\overline{\lambda}_0$ assigns no mass to the borders $M_1 := \{0\} \times (-\infty,K)$ and $ M_2 :=  \{T\} \times (-\infty,K)$. Indeed, in this case we have
\begin{align*}
\int_\Omega \widetilde{g}(x)   \lambda_0(d(t,x))
= \int_{\cl\,\Omega} \widetilde{g}(x)   d\overline{\lambda}_0(t,x)
= \lim_{n\to\infty}  \int \widetilde{g}(x)   \lambda_{\epsilon_n}(d(t,x))
= \lim_{n\to\infty}  d_{\epsilon_n} = d_0.
\end{align*}
The second equality follows from the weak convergence of the sequence $\lambda_{\epsilon_n}$ in the space $\M(\cl\,\Omega)$ and the boundedness of the continuous function $\widetilde{g}$.
The last equality has already been established in \eqref{EQ_strong_duality_p0_d0}.

Assume by contradiction that $\overline{\lambda}_0$ assigns mass to the set $M_1$.
In this case there is some  $\alpha < K$ with $\overline{\lambda}_0 \left( \{0\} \times [\alpha,K) \right) > 0$.
Due to weak convergence in the space $\M\!\left( [0,T/2] \times [\alpha,K] \right)$ we have
\begin{align*}
\int_{\{0\}  \times [\alpha,K) } \kappa(t,x,y)    \overline{\lambda}_0(d(t,x))
&\leq  \int_{ [0,T/2] \times [\alpha,K] }  \kappa(t,x,y)   \overline{\lambda}_0(d(t,x))\\
&=  \lim_{n\to\infty} \int_{ [0,T/2] \times [\alpha,K] } \kappa(t,x,y)    \lambda_{\epsilon_n}(d(t,x)) \leq  1
\end{align*}
for any $ y \in (-\infty,K]$. Fatou's lemma, Lemma \ref{LEMMA_norm_and_quotient_of_normal_pdf}(3), and $K < x_0$ now yield the following contradiction
\begin{align*}
1 &\geq \liminf_{y \to -\infty} \int_{\{0\} \times [\alpha,K) } \frac{ \nno\!\left(x +\widetilde{r} T, \sigma^2 T, y \right)}{  \nno\!\left(x_0 +\widetilde{r} T, \sigma^2 T, y \right) }  \overline{\lambda}_0(d(t,x))\\
&\geq  \int_{\{0\}  \times [\alpha,K) } \liminf_{y \to -\infty} 
\exp\!\left( y \frac{x-x_0}{\sigma^2 T}   \right)
\exp\!\left( \frac{x_0^2 - x^2 +  2 \widetilde{r}T(x_0-x)}{2\sigma^2 T} \right)
  \overline{\lambda}_0(d(t,x))= \infty
\end{align*}
Hence the assumption was wrong and therefore $\overline{\lambda}_0(M_1)  = 0$.

Next we turn our attention to the set $M_2$. 
For any $(t,x)\in [0,T]\times\rr $ we define by
\begin{equation}\label{DEF_lsc_continuation_of_T_mu}
V(t,x):= \liminf_{\substack{(t',x') \to (t,x)\\(t',x') \in (0,T)\times \rr} } \T\mu_0 (t',x')
\end{equation}
the lower semi-continuous extension of the function $ \T\mu_0$ to the set $[0,T]\times\rr$.
We  show that imposing the \emph{additional assumption}
\begin{equation}\label{EQ_ADDITIONAL_ASSSUMPTION_on_V}
V(T,x) > \widetilde{g}(x) \quad \forall x \in (-\infty,K)
\end{equation}
warrants that the measure $\overline{\lambda}_0$ assigns no mass to the set $M_2$.
The mapping $V$ is lower semi-continuous and bounded from below and hence attains its minimum on any compact subset of $[0,T]\times\rr$.
Moreover, assumption \eqref{EQ_ADDITIONAL_ASSSUMPTION_on_V} ensures that the minimal value of the function $V - \widetilde{g}$  is strictly positive
on any set $\{T\} \times [a,b] \subset  M_2$ with $a < b < K$.
By lower semi-continuity there is some $n_0 \in\nn$ and $\delta > 0$ such that
\begin{equation}\label{EQUATION_delta_bound_for_difference}
 V(t,x) - \widetilde{g}(x) \geq \delta  
\end{equation}
for any $(t,x)\in [T-1/n_0, T] \times [a,b]$. 
Assume by contradiction that the measure $\overline{\lambda}_0$ assigns mass to $M_2$.
We can choose some strip $\{T\} \times (a,b) \subset M_2$ and a constant $\rho>0$ such that 
$\overline{\lambda}_0(Q_m) \geq 2 \rho$ holds for any $m\in\nn$, where $Q_m := (T-1/m,T] \times (a,b)$. 
The measures $\lambda_{\epsilon_n}$ converge weakly in $\M(\cl\,\Omega)$ to $\overline{\lambda}_0$.
Owing to \cite[Theorem 13.16]{klenke.14},
we can pass to a subsequence (again denoted by $\epsilon_n$) such that 
$\lambda_{\epsilon_n}(\inner Q_n) \geq \rho$ for all $n\in\nn$. The strong duality in the \ref{PROGRAM_D_eps}-\ref{PROGRAM_P_eps}-setting yields
\begin{align*}
\langle \T\mu_0 - \widetilde{g} , \lambda_{\epsilon_n} \rangle 
&= \langle \mu_0, \T^*\lambda_{\epsilon_n} \rangle - p_{\epsilon_n} \\
&= \langle \mu_0, \T^*\lambda_{\epsilon_n} - 1\rangle + \Vert \mu_0 \Vert- \Vert \mu_{\epsilon_n} \Vert   \\
&\leq \Vert \mu_0 \Vert- \Vert \mu_{\epsilon_n} \Vert.
\end{align*}
Moreover, \eqref{EQUATION_delta_bound_for_difference} implies
\begin{align*}
\langle \T\mu_0 - \widetilde{g} , \lambda_{\epsilon_n} \rangle 
&\geq \int_{\inner Q_n} V(t,x) - \widetilde{g}(x)   \lambda_{\epsilon_n}(d(t,x))\\
&\geq \delta \lambda_{\epsilon_n}(\inner Q_n)\\
&\geq \delta \rho > 0,\quad n\geq n_0.
\end{align*}
However, we already know from \eqref{EQ_strong_duality_p0_d0} that
$\Vert \mu_0 \Vert- \Vert \mu_{\epsilon_n} \Vert \to 0 $ as $n\to \infty$. 
This yields a contradiction, which finally shows that $\overline{\lambda}_0 (M_2) = 0$.

Last but not least, we observe that literally the same calculation as in \eqref{COMPLEMENTARY_SLACKNESS_Peps_Deps} yields the complementary slackness property 
for $\mu_0$ and $\lambda_0$ in the case of primal and dual attainment. 
This means that the equation
\begin{equation}\label{EQ_CS_lambda}
\int  \kappa(t,x,y)   d\mu_0(y) = \widetilde{g}(x)
\end{equation}
holds $\lambda_0$-a.e.\ on $\Omega$ and 
\begin{equation}\label{EQ_CS_mu}
\int_\Omega  \kappa(t,x,y)   \lambda_0(d(t,x)) = 1
\end{equation}
holds $\mu_0$-a.e.\ on $(-\infty,K]$. Let us summarise our results from above.
\begin{lemma}\label{Theorem_summary}  
\begin{enumerate}
\item For any $\epsilon \in (0,T/2)$ the linear programmes \ref{PROGRAM_P_eps} and \ref{PROGRAM_D_eps} have solutions $\mu_\epsilon$ and $\lambda_\epsilon$.  
The optimal values $p_\epsilon$ and $d_\epsilon$ of the latter programmes coincide.  The total mass of the optimisers is bounded by some constant $\rho \in (0,\infty) $ which
does not depend on $\epsilon$. Moreover, the measure $\lambda_\epsilon$ assigns no mass to the zero set of the function $\widetilde{g}$. 
The equation $\T \mu_\epsilon = \widetilde{g}$ holds $\lambda_\epsilon$-a.e.\ on $\Omega_\epsilon$ and $\T^* \lambda_\epsilon  = 1$ holds $\mu_\epsilon$-a.e.\ on $(-\infty,K]$.
\item There exists a sequence $\epsilon_n \downarrow 0$ such that $\mu_{\epsilon_n}$ converges vaguely in $\M(-\infty,K]$ to some \ref{PROGRAM_P_0_functional_analytic}-admissible measure $\mu_0$ and 
$\lambda_{\epsilon_n}$ converges vaguely in $\M(\Omega)$ to some \ref{PROGRAM_D_0_functional_analytic}-admissible measure $\lambda_0$. The optimal value of \ref{PROGRAM_P_0_functional_analytic} is obtained by $\mu_0$  and coincides
with the optimal value of \ref{PROGRAM_D_0_functional_analytic}. The measure $\mu_0$ assigns mass to any open subset of $(-\infty, K)$ and $\Vert\mu_0 \Vert \vee \Vert\lambda_0 \Vert \leq \rho$.
\item Let $V$ be defined as in \eqref{DEF_lsc_continuation_of_T_mu}. If $V(T,x) > \widetilde{g}(x)$ for any $x \in (-\infty,K)$,
the optimal value of the programme \ref{PROGRAM_D_0_functional_analytic} is obtained by $\lambda_0$. In this case the complementary slackness 
equations \eqref{EQ_CS_lambda} and \eqref{EQ_CS_mu} hold.
\end{enumerate}
\end{lemma}
Lemma \ref{Theorem_summary1} is nothing but a slight reformulation of statements 2 and 3.

\subsection{Proof of Theorem \ref{t:THEOREM_Hauptsatz}}\label{section_european_representation}
We use the notation from the preceding sections. In particular, see Section \ref{SubSection_DUALITY} for the definition of the operator $\T$ and the optimisation problem
\ref{PROGRAM_P_0_functional_analytic}. Let $\mu^*$ be a cheapest dominating European option in the sense of Theorem \ref{t:Theorem_summary0}.
In view of \eqref{EQ_TRANSFORM_BEIBEL_LERCHE} we have
\begin{align}\label{eq:Veu_in_T_umschreiben}
v_{\eu, \mu^*}(\theta, x)
&= e^{-r\theta}\int   \frac{ \nno\!\left(x +\widehat{r} \theta, \sigma^2 \theta, y \right)}{ \nno\!\left(x_0 +\widehat{r} T, \sigma^2 T, y \right)}    \mu^*(dy)\\
&= e^{(-2r/\sigma^2) x} \int   \frac{ \nno\!\left(x +\widetilde{r} \theta, \sigma^2 \theta, y \right)}{ \nno\!\left(x_0 +\widetilde{r} T, \sigma^2 T, y \right)}   e^{(2r / \sigma^2)x_0 - rT}   \mu^*(dy)\nonumber\\
&= e^{(-2r/\sigma^2) x}  \T\mu_0(T-\theta,x)\nonumber
\end{align}
for any $(\theta,x)\in (0,T)\times\rr$. Here we denote by $\mu_0 =  e^{(2r / \sigma^2)x_0 - rT}  \mu^*$ the 
corresponding \ref{PROGRAM_P_0_functional_analytic}-optimal measure from Lemma \ref{Theorem_summary1}.\\

\paragraph{Step 1: Analyticity of the European value function}
First, we show that the first assumption of Theorem \ref{t:THEOREM_Hauptsatz} ensures the analyticity of the function
$v_{\eu,\mu^*}$ on the open $\cc^2$-domain 
\[E:= \left\{ \theta \in \cc  :  \sqrt{ (\Re \theta - (T+2\delta)/2)^2 + (\Im \theta)^2 } < (T+2\delta)/2  \right\} \times \cc. \]
It suffices to verify that the function
\begin{equation}\label{EQ_Function_for_partial_analyticity}
e^{r\theta}v_{\eu, \mu^*}(\theta, x) = \int
 \frac{ \nno\!\left(x +\widehat{r} \theta, \sigma^2 \theta, y \right)}{ \nno\!\left(x_1 +\widehat{r} (T+2\delta), \sigma^2 (T+2\delta), y \right)}   \mu^{**}(dy)
\end{equation}
is analytic on $E$, where 
\[\frac{ d\mu^{**}}{ d\mu^*}(y) := \frac{ \nno\!\left(x_1 +\widehat{r} (T+2\delta), \sigma^2 (T+2\delta), y \right)}{ \nno\!\left(x_0 +\widehat{r} T, \sigma^2 T, y \right)} ,\quad y\in\rr.\]
In view of assumption 1 we have $\Vert \mu^{**} \Vert = v_{\eu, \mu^*}(T+2\delta, x_1)e^{r(T+2\delta)}< \infty$.
Due to Hartogs' theorem it is enough to show that the function from \eqref{EQ_Function_for_partial_analyticity} is partially analytic, cf.\ \cite[Paragraph 2.4]{krantz.92}.
Lemma \ref{LEMMA_norm_and_quotient_of_normal_pdf}(2) implies that 
\begin{align}
\MoveEqLeft{\left\vert \frac{\nno(x+\widehat{r}\theta, \sigma^2 \theta,y)}{\nno(x_1 + \widehat{r}(T+2\delta), \sigma^2 (T+2\delta),y)} \right\vert = \left\vert h_1(\theta,x)\right\vert \left\vert \exp\biggl( -\frac{(y - A(\theta,x))^2}{2B(\theta,x)} \biggr)\right\vert}\nonumber\\
&=\left\vert h_2(\theta,x)\right\vert \exp\!\left( - \frac{\Re B(\theta,x)}{2\vert B(\theta,x)\vert^2} \biggl( y -  \Re A(\theta,x) - \frac{ \Im A(\theta,x) \Im B(\theta,x) }{\Re B(\theta,x)}  \biggr)^2 \right) \nonumber
\end{align}
for any $(\theta,x)\in E$ and $y\in\rr$,
where $h_1, h_2$ denote certain functions which are continuous on $E$.
The quantities $A(\theta,x)$ and $B(\theta,x)$ are defined as in \eqref{e:AB}.
For any $(\theta,x)\in E$ we have 
\[\Re B(\theta,x)
=  \Re \frac{\sigma^2\theta(T+2\delta)}{T+2\delta-\theta} 
= \frac{\sigma^2 (T+2\delta)}{ |T+2\delta - \theta|^2}    \left( (T+2\delta) \Re \theta - |\theta|^2  \right) > 0 \]
and therefore the integrand occurring in \eqref{EQ_Function_for_partial_analyticity} satisfies the inequality  
\[\sup_{y\in\rr}\left\vert \frac{\nno(x+\widehat{r}\theta, \sigma^2 \theta,y)}{\nno(x_1 + \widehat{r}(T+2\delta), \sigma^2 (T+2\delta),y)} \right\vert \leq \left\vert h_2(\theta,x)\right\vert.\]
The quantity on the right-hand side is bounded on every compact subset of $E$.
Hence we can use a standard argument involving the theorems of Morera and Fubini in order to prove partial analyticity.
For a detailed exposition of the technique, we refer the reader to the proof of Lemma \ref{LEMMA_analyticity_sol_heat_eaqtion}, which is to be found in \cite[Section 5.1]{lenga.17}.
In view of Hartogs' theorem we conclude that the mapping $v_{\eu,\mu^*}$ is indeed analytic on $E$.

\paragraph{Step 2: Analyticity of the curve}
We show that the curve $\theta \mapsto \cx(\theta)$ is analytic on an open complex domain containing the interval $(0,T]$. 
In view of step 1 and the assumptions imposed on the American payoff $g$,  the function 
\begin{equation}\label{EQ_def_func_Psi_from_step_2}
 \Psi(\theta,x):= v_{\eu,\mu^*}(\theta,x) - g(x)
\end{equation}
is analytic on the open $\cc^2$-domain $D' \times D$, where
\begin{align*}
D'&:= \left\{ \theta \in \cc : \sqrt{ (\Re \theta - (T+2\delta)/2)^2 + (\Im \theta)^2 } < (T+2\delta)/2  \right\} 
\end{align*}
and $D$ denotes the domain of analyticity of $g$.
The set $D'$ is simply connected and $(0,T+2\delta) \times (-\infty,K)$ is a subset of $D' \times D$ .
The continuity of $\Psi$  and the uniqueness of the minima warrant that the curve $x$ is continuous on the interval $(0, T + \delta/2)$.
Indeed, assume by contradiction that $\cx$ is discontinuous at $\theta_0 \in (0, T + \delta/2)$. 
Then there is a sequence $\theta_n \to \theta_0$ and some $x_\infty\leq K,\epsilon >0$ such that 
 $\cx(\theta_n) \to x_\infty$ as $n \to \infty$ and $\vert x_\infty - \cx(\theta_0)\vert > \epsilon$. 
Hence there exists a constant $T_{\max} >0$ with
$\Psi(\theta_0,\cx(\theta_0)) + T_{\max} < \Psi(\theta_0,x_\infty)$. Consequently, we can choose two disjoint balls $B((\theta_0,x_\infty),r)$
and $B((\theta_0,\cx(\theta_0)),r)$ of radius $ r \in (0, \epsilon /2)$ with
$\Psi(\theta,x) + {T_{\max}}/{2} < \Psi(\widetilde\theta,\widetilde x)$
for any $(\theta,x) \in B_r{(\theta_0,\cx(\theta_0))} $ and $(\widetilde\theta,\widetilde x) \in B((\theta_0,x_\infty),r)$.
This yields a contradiction because $(\theta_n, \cx(\theta_n))$ is contained in $B((\theta_0,x_\infty),r)$ for any sufficiently large integer $n$.
Hence the curve $\cx$ is continuous.

For any $\theta\in (0,T+\delta)$ we have 
$D_2\Psi(\theta,\cx(\theta)) = 0$ and $D_{22} \Psi (\theta,\cx(\theta)) \geq 0$ due to the necessary first and second order conditions for minimality. 
Applying Kolmogorov's backward equation in the version of Lemma \ref{l:kolmogorov}
we obtain
\begin{equation}\label{EQ_express_Psi_xx_in_therms_of_H}
\begin{aligned}
D_{22} \Psi  &= D_{22} v_{\eu,\mu^*} - g'' \\
	&= \frac{2}{\sigma^2}D_1v_{\eu,\mu^*} + \left(1 - \frac{2r}{\sigma^2} \right)D_2 v_{\eu,\mu^*} + \frac{2r}{\sigma^2} v_{\eu,\mu^*} - g''\\
	&= \frac{2}{\sigma^2}D_1 v_{\eu,\mu^*} + \left(1 - \frac{2r}{\sigma^2} \right)D_2\Psi + \frac{2r}{\sigma^2} \left( v_{\eu,\mu^*} - g\right)- c\\
	&= H + \left(1 - \frac{2r}{\sigma^2} \right)D_2 \Psi 
\end{aligned}
\end{equation}
on $(0,T+2\delta) \times \rr$, where $c$ and $H$ are defined as in \eqref{eq:function_c} and \eqref{EQ_function_H_of_verification_theorem}, respectively. 
Assumption 3 warrants that $H$ and therefore $ D_{22} \Psi $ are strictly positive on the set $ \Gamma:= \{ (\theta , \cx(\theta)) : \theta \in (0,T] \}$.
Therefore Theorem \ref{THEOREM_Implicit_Function} is applicable to the function $D_2\Psi$ at any point of $\Gamma$.
We obtain that for any $(\widetilde{\theta}, \widetilde{x}) \in \Gamma$ there exist open neighbourhoods $\widetilde{\theta} \in U_{\widetilde{\theta}}, \widetilde{x}\in U_{\widetilde{x}}$ and 
an analytic curve $\chi_{\widetilde{\theta}} : U_{\widetilde{\theta}} \to U_{\widetilde{x}}$ with
$\chi_{\widetilde{\theta}}(\theta) = \cx(\theta)$
for any $\theta \in  U_{\widetilde{\theta}}\cap (0,T]$.
The identity theorem implies that any two curves $\chi_{\widetilde{\theta}_1}, \chi_{\widetilde{\theta}_2}$  coincide on $U_{\widetilde{\theta}_1} \cap U_{\widetilde{\theta}_2}$.
Since the mapping $\theta \mapsto \cx(\theta)$ is continuous, there exists
an analytic function $\chi$ such that  
$\chi \vert_{U_\theta} = \chi_\theta$ 
for any $ \theta \in (0,T]$. In particular, we have $\chi(\theta) = \cx(\theta)$ for any $ \theta \in (0,T]$. 
This proves that $x$ is indeed analytic on some complex domain containing the interval $(0,T]$.

\paragraph{Step 3: Proof of statement 2}
We verify that $ v_{\eu, \mu^*}(\theta,\cx(\theta)) = g(\cx(\theta))$ for any $\theta \in [0,T]$.
Since the measure $\mu^*$ assigns no mass to the set $(K,\infty)$, we have $v_{\eu, \mu^*}(0,x) = 0$ for any $x>K$.
Lower semi-continuity even implies  $v_{\eu, \mu^*}(0,K) = 0$ and therefore
\[v_{\eu, \mu^*}(0,x(0)) = v_{\eu, \mu^*}(0,K) = 0 = g(K) =  g(\cx(0)).\]
In view of \eqref{eq:Veu_in_T_umschreiben} we have 
\[ e^{(2r/\sigma^2)x}\left(  v_{\eu, \mu^*}(T-t,x) - g(x) \right) = \T\mu_0(t,x) - \widetilde{g}(x) \]
for any $(t,x) \in [0,T)\times\rr$, where $\widetilde{g}$ is defined as in \eqref{EQ_DEFINITION_of_tilde_g}. 
Assumption 2 implies that $ v_{\eu, \mu^*}(0,x) - g(x) > 0 $ for any $x <K$.
Consequently, Lemma \ref{Theorem_summary} warrants strong duality, primal and dual attainment as well as complementary slackness.
In view of \eqref{EQ_CS_lambda}, the 
dual maximiser $\lambda_0$ assigns no mass to the complement of the set $\{ (t,\cx(T-t)) :  0 < t < T \}$. 
We claim that there exists a sequence  $ \theta_n \uparrow T$ with $\theta_n \in (0,T)$ and
\begin{equation}\label{eq:sequence_with_desired_property}
v_{\eu, \mu^*}(\theta_n,\cx(\theta_n)) = g(\cx(\theta_n)), \quad n\in\nn.
\end{equation}
Assume to the contrary that this is false.
Then there is some $\epsilon \in (0,T)$ with $v_{\eu, \mu^*}(\theta,\cx(\theta)) > g(\cx(\theta))$ for all $\theta \in (T-\epsilon, T)$.
Equation \eqref{EQ_CS_lambda} tells us that the measure $\lambda_0$ is concentrated on the set $\Gamma_\epsilon := \{  (t,\cx(T-t))  : \epsilon < t < T \}$.
From Lemma \ref{Theorem_summary} we already know that the primal minimiser $\mu_0$ assigns mass to any open subset of $(-\infty,K)$. 
By \eqref{EQ_CS_mu} we can find a sequence $y_n \downarrow - \infty$ 
with $\max_{n\in\nn} y_n < \min_{\theta\in [0,T]} \cx(\theta) + \widetilde{r}(T-\epsilon) =: z$ and
\begin{align*}
\nno\bigl(x_0 + \widetilde{r} T, \sigma^2 T, y_n \bigr)  
&= \int_{\Gamma_\epsilon} \nno\bigl(x + \widetilde{r} (T-t), \sigma^2 (T-t), y_n \bigr)   \lambda_0(d(t,x)), \quad n\in\nn.
\end{align*}
In view of $\widetilde{r}< 0$ we have 
\begin{align*}
\nno\bigl(x_0 + \widetilde{r} T, \sigma^2 T, y_n \bigr)  
&\leq \int_{\Gamma_\epsilon} \nno\bigl( z, \sigma^2 (T-t), y_n \bigr)   \lambda_0(d(t,x)) \\
&\leq \nno\bigl( z, \sigma^2 (T-\epsilon), y_n \bigr)  \lambda_0\bigl(\Gamma_\epsilon\bigr), \quad n\in\nn.
\end{align*}
This yields the contradiction
\begin{align*}
1 &\leq  \lambda_0\!\left(\Gamma\right)  \lim_{n \to\infty}  
\frac{ \nno\bigl( z, \sigma^2 (T-\epsilon), y_n \bigr) }{\nno\bigl(x_0 + \widetilde{r} T, \sigma^2 T, y_n \bigr)  } = 0.
\end{align*}
Consequently a sequence with the desired property \eqref{eq:sequence_with_desired_property} exists.

In view of steps 1 and 2, the mapping $\theta \mapsto v_{\eu, \mu^*}(\theta,\cx(\theta)) - g(\cx(\theta))$ is 
analytic on some open complex domain
containing the interval $(0,T]$. Equation \eqref{eq:sequence_with_desired_property} and the identity theorem finally yield that 
$v_{\eu, \mu^*}(\theta,\cx(\theta)) = g(\cx(\theta))$ for any $\theta\in(0,T]$.

\paragraph{Step 4: Proof of statement 4}
We verify that $\mu^*$ is the unique measure representing our American payoff on the set $\widetilde{C}_{T,x_0}$ as defined in \eqref{EQ_set_C_tilde_from_verification_theorem}.
Moreover, we show that $\widetilde{C}_{T,x_0}$ is a connected subset of $C_{T,x_0}$.
For any $T_0 \in [0,T]$ the process 
$V_t^{(T_0)}:= e^{-rt} v_{\eu, \mu^*}({T_0}-t, X_t)$ is a martingale on the interval $[0,{T_0})$.
Indeed, for $0 \leq u < t + u < {T_0}$ the Markov property of the process $X$ yields 
\begin{equation}\label{EQ_martingale_property_of_discounted_measure_EU_value_process}
\begin{aligned}
\E_x\bigl( V_{t+u}^{(T_0)} \big| \mathscr{F}_u \bigr)
&=  e^{-r(t+u)}  \E_{X_u}\!\left(  v_{\eu, \mu^*}({T_0}-t-u, X_{t}) \right)\\ 
&=  e^{-r{T_0}} \int    \frac{ \E_{X_u}\!\left(  \nno\!\left(  X_t + \widehat{r}({T_0}-t-u) , \sigma^2 ({T_0}-t-u), y \right) \right)}
{ \nno\!\left(x_0 + \widehat{r} T, \sigma^2 T, y \right)} \mu^*(dy) \\ 
&=   e^{-r{T_0}} \int    \frac{ \nno\!\left(  X_u + \widehat{r}({T_0}-u) , \sigma^2 ({T_0}-u), y \right) }{ \nno\!\left(x_0 + \widehat{r} T, \sigma^2 T, y \right)}  \mu^*(dy) \\
& = e^{-ru}   v_{\eu, \mu^*}({T_0}-u, X_{u}) =  V_u^{(T_0)}. 
\end{aligned}
\end{equation}
The third equality follows from the convolution property of the normal distribution.

The martingale condition may fail to hold up to $T_0$. Nevertheless, Fatou's lemma yields the supermartingale property. Indeed, for any $u \in [0,{T_0}]$ we have 
\begin{align*}
\E_x \bigl( V_{T_0}^{(T_0)} \big| \mathscr{F}_u \bigr) 
&=  \E_{x}\!\left(   e^{-r{T_0}} \liminf_{ t \uparrow {T_0} } v_{\eu, \mu^*}({T_0}-t, X_t)\middle\vert \mathscr{F}_u \right) \\ 
&\leq  \liminf_{t \uparrow {T_0}  } \E_{x} \!\left(   e^{-rt}  v_{\eu, \mu^*}({T_0}-t, X_t)\middle\vert \mathscr{F}_u  \right)  = V_u^{(T_0)}.
\end{align*}

Due to superreplication, we have  $ e^{-rt} g(X_t) \leq V_t^{(T)} $ for any $t \in [0,T]$ and consequently the optional sampling theorem yields
\[ \E_x \!\left( e^{-r\tau} g(X_\tau)  \middle\vert \mathscr{F}_t \right)
\leq \E_x \bigl( V_\tau^T \big\vert \mathscr{F}_t \bigr) \leq  V_t^{(T)}\]
for any $[t,T]$-valued stopping time $ \tau$.
Maximising the left-hand side over all such stopping times shows that $v_{\am,g}(\theta,x) \leq v_{\eu,\mu^*}(\theta,x) $ for any $(\theta,x)\in [0,T] \times \rr$.

Now we verify that the value functions $v_{\am,g}$ and $v_{\eu,\mu^*}$ coincide on the set $\widetilde C_{T,x_0}$.
To this end let $\tau_\theta$ be defined as in \eqref{EQ_optimal_stopping_time_from_Hauptsatz}.
Assumption 4 warrants that the measure $\mu^*$ has no atom at $K$, i.e.\ $\mu^*(\{K\})=0$. Indeed, assuming $\mu^*(\{K\})>0$ would imply that
\[\liminf_{\theta\to 0} v_{\eu, \mu^*}(\theta ,K)
\geq T_{\max}  \liminf_{\theta\to 0} e^{-r\theta} \nno\bigl(\widehat{r} \theta, \sigma^2 \theta, 0 \bigr) = \infty, \]
where $T_{\max}$ denotes some positive constant. 
Owing to the geometric properties of the curve $\cx$, we have 
\[ \E_{x}\!\left( \nno\bigl(X_{\tau_\theta}  +\widehat{r} (\theta - \tau_\theta), \sigma^2 (\theta - \tau_\theta), y \bigr)  1_{\{ \tau_\theta = \theta\}} \right) 
\leq  \E_{x}\!\left( 1_{\{y\}}(X_{\theta})  1_{\{ X_{\theta} \geq K \}} \right) 
= 0 \]
for any $y < K$. Hence monotone convergence yields
\begin{align}\label{e:sechs}
v_{\am,g}(\theta,x) 
&\geq \E_{x} \!\left( e^{-r \tau_\theta}  g(X_{\tau_\theta}) \right)  \\
&\geq \E_{x} \!\left( e^{-r \tau_\theta}  g(X_{\tau_\theta})  1_{\{ \tau_\theta < \theta\}}   \right)  \nonumber \\
&=  \E_{x}\!\left( e^{-r \tau_\theta}  v_{\eu, \mu^*}(\theta - \tau_\theta, X_{\tau_\theta} )  1_{\{ \tau_\theta < \theta\}   } \right) \nonumber\\
&= \lim_{x' \uparrow K} \E_{x}\!\left( e^{-r\theta}
\int_{-\infty}^{x'}   \frac{ \nno\!\left(X_{\tau_\theta}  +\widehat{r} (\theta - \tau_\theta), \sigma^2 (\theta - \tau_\theta), y \right)}{ \nno\!\left(x_0 +\widehat{r} T, \sigma^2 T, y \right)}    \mu^*(dy)
 1_{\{ \tau_\theta < \theta\}}  \right) \nonumber\\
&= \lim_{x' \uparrow K}  e^{-r\theta} \int_{-\infty}^{x'}  
 \frac{ \E_{x}\!\left( \nno\!\left(X_{\tau_\theta}  +\widehat{r} (\theta - \tau_\theta), \sigma^2 (\theta - \tau_\theta), y \right)  \right)}{ \nno\!\left(x_0 +\widehat{r} T, \sigma^2 T, y \right)}    \mu^*(dy) \nonumber\\
&=  v_{\eu, \mu^*}(\theta ,x )\nonumber\\
&\geq v_{\am,g}(\theta,x) \label{e:sechs2}
\end{align}
for any $(\theta, x) \in \widetilde C_{T,x_0}$.
Summing up, we have shown that $v_{\am,g}(\theta,x) = v_{\eu, \mu^*}(\theta, x) > g(x)$ holds for any $(\theta,x)\in \widetilde C_{T,x_0}$.
Moreover, this directly implies that $\widetilde C_{T,x_0}$ is a connected subset of $C_{T,x_0}$.

Finally we verify that the representing measure $\mu^*$ is unique. Assume that there is another measure $\nu$ such that
$v_{\eu, \mu^*}(\theta,x) = v_{\am,g}(\theta,x) = v_{\eu, \nu}(\theta,x)$ for any $(\theta,x) \in  \widetilde C_{T,x_0}$.
Recall that the value functions $v_{\eu, \mu^*}, v_{\eu, \nu}$ are analytic on a $\cc^2$-domain containing the set $(0,T)\times\rr$.
The set $\widetilde C_{T,x_0}$ contains some open ball.
By applying the identity theorem in each variable, we conclude that the mappings $v_{\eu, \mu^*}$ and $v_{\eu, \nu}(\theta,x)$ 
coincide on the set $(0,T)\times\rr$.
Equation \eqref{eq:Veu_in_T_umschreiben} implies that $\T\mu^* = \T\nu$ holds on $(0,T)\times\rr$. 
By Lemma \ref{LEMMA_injectivity_of_operator_T} the operator $\T$ is injective on the Borel measures and therefore $\mu^* = \nu$.

\paragraph{Step 5: Proof of statement 5}
By statement 2 we have 
\[C_{T,x_0}=\bigl(C_{T,x_0}\cap\{(\theta,x)\in(0,T]\times\rr:\cx(\theta)<x\}\bigr)
\cup\bigl(C_{T,x_0}\cap\{(\theta,x)\in(0,T]\times\rr:\cx(\theta)>x\}\bigr).\]
Since $C_{T,x_0}$ is connected and the first set in the union is nonempty, the second set must be empty.
Therefore $C_{T,x_0}\subset \widetilde C_{T,x_0}$ and hence $C_{T,x_0}=\widetilde C_{T,x_0}$ by step 4.

\paragraph{Step 6: Proof of statement 3}
For $x \geq K$ we have $v_{\eu, \mu^*}(0,x) =  g(x) =  0$.  
For any $x \in [\min_{\theta\in (0,T]} \cx(\theta),K)$ there is a maturity $\theta(x) \in (0,T]$
such that $(\theta(x),x)$ is located on the curve, i.e.\ $\cx(\theta(x)) = x$. Due to superreplication and assertion 2, we have
\begin{equation}\label{e:fuer6}
 g(x) \leq \inf_{\theta\in[0,T]} v_{\eu, \mu^*}(\theta,x)  \leq v_{\eu, \mu^*}(\theta(x),x) = g(x).
\end{equation}

\paragraph{Step 7: Proof of statement 6}
This follows from (\ref{e:sechs}, \ref{e:sechs2}).

\paragraph{Step 8: Monotonicity of the curve}
We show by contradiction that $\cx$ is strictly increasing. 
First assume that it is not increasing. Since $\cx$ is continuous,
there are $0 < \theta_0 < \theta_1 < \theta_2 < T$ and $x_0 < K$ such that
$\cx(\theta_0) = \cx(\theta_2) = x_0$ and  $\cx(\theta_1) < x_0$. 
From step 4 we know that $\widetilde C_{T,x_0}$ is a connected subset of $C_{T,x_0}$.
In particular, $(\theta_1,x_0)$ is located within the continuation set.
In view of assertion 2 we  conclude that
$g(x_0) =  v_{\am, g}(\theta_0,x_0) < v_{\am, g}(\theta_1,x_0) \leq v_{\am, g}(\theta_2,x_0) = g(x_0)$. 
This is impossible and hence the mapping $\theta\mapsto \cx(\theta)$ must be increasing.

Now assume that there are some $0 <  \theta_0 < \theta_1 \leq T$ such that $x$ is constant on the interval $(\theta_0, \theta_1)$.
Since the curve $\cx$ is analytic, the identity theorem implies that $\cx$ is constant on $(0,T]$. 
In view of the second assumption we find that $K > \cx(\theta) = \liminf_{\theta' \to 0} \cx(\theta') = K$ for any $\theta \in (0,T]$,
which is impossible.
Therefore $\theta\mapsto \cx(\theta)$ is strictly increasing. 
\qed

\section{Representability of the American put}\label{s:SEC_Example_am_put}
\begin{figure}
\centering
\includegraphics[width=.92\textwidth]{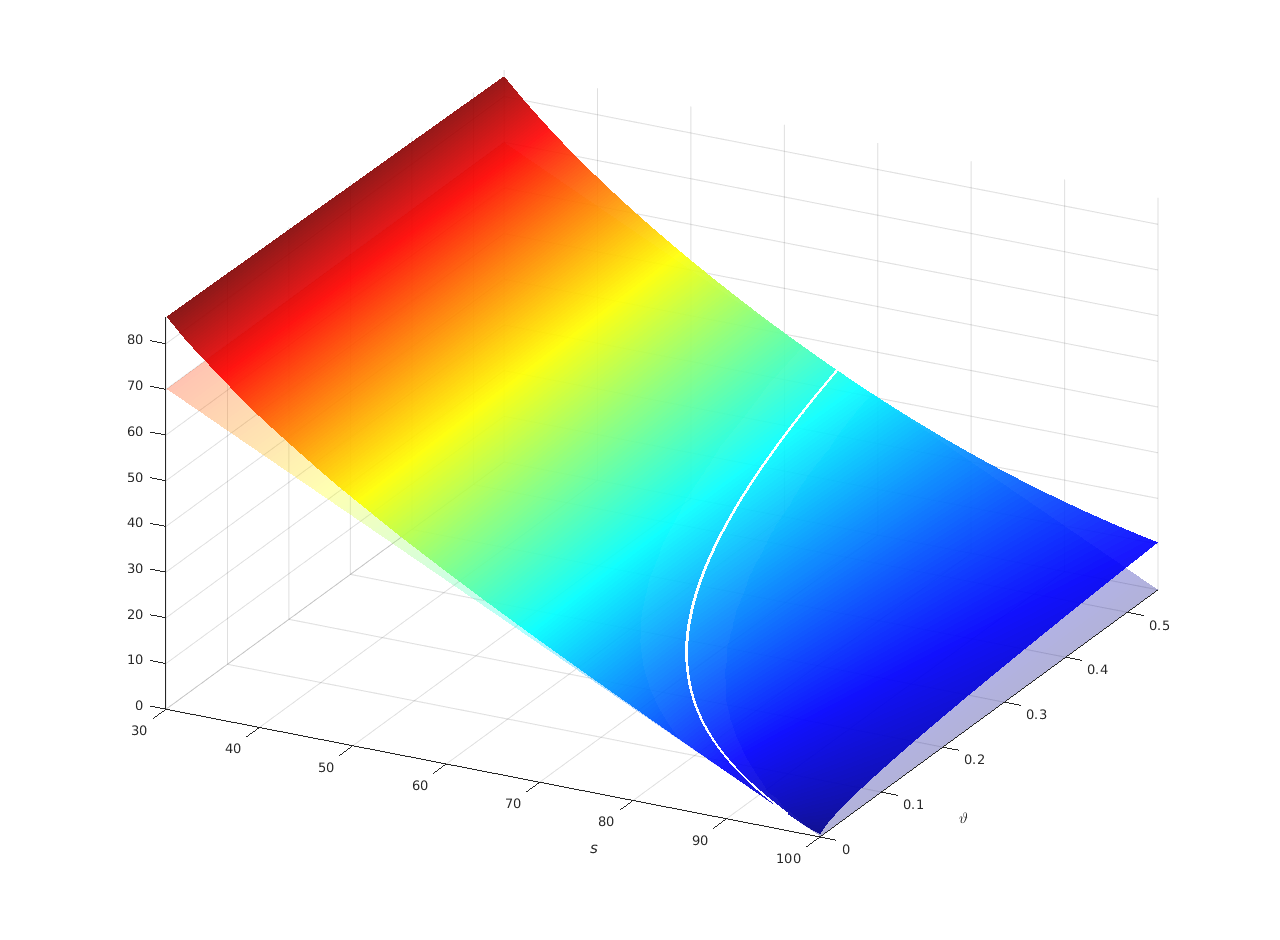}
\caption{The price surface of the CDEO associated to the American put}\label{f:FIGURE_CDEO_PUT}
\end{figure}
\begin{figure}
\centering
\includegraphics[width=.9\textwidth]{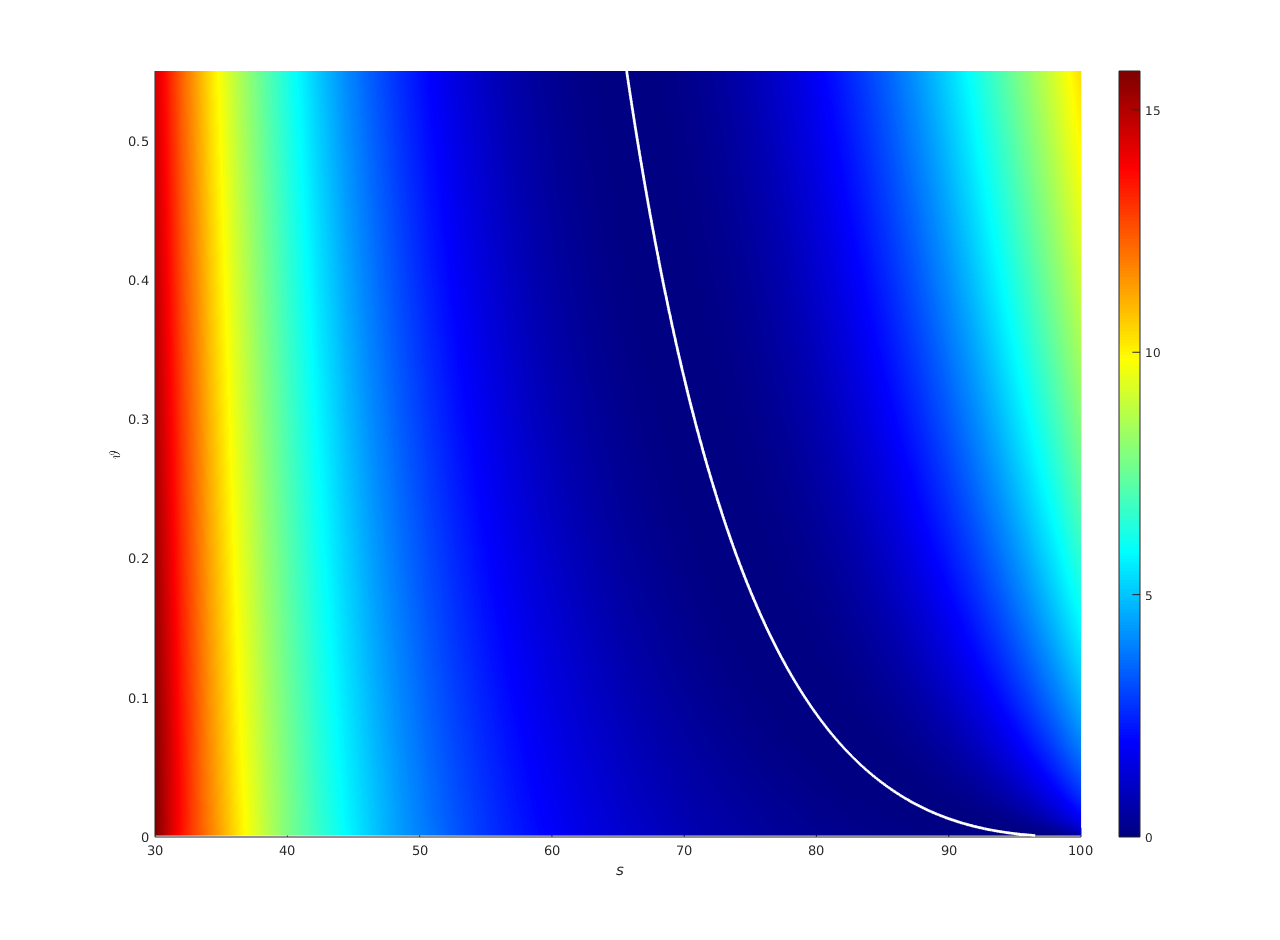}
\caption{The curve $\theta\mapsto \exp \cx(\theta)$ for the CDEO of the American put}\label{f:FIGURE_AMPUT_MIN}
\end{figure}
\begin{figure}
\centering
\includegraphics[width=0.66\textwidth]{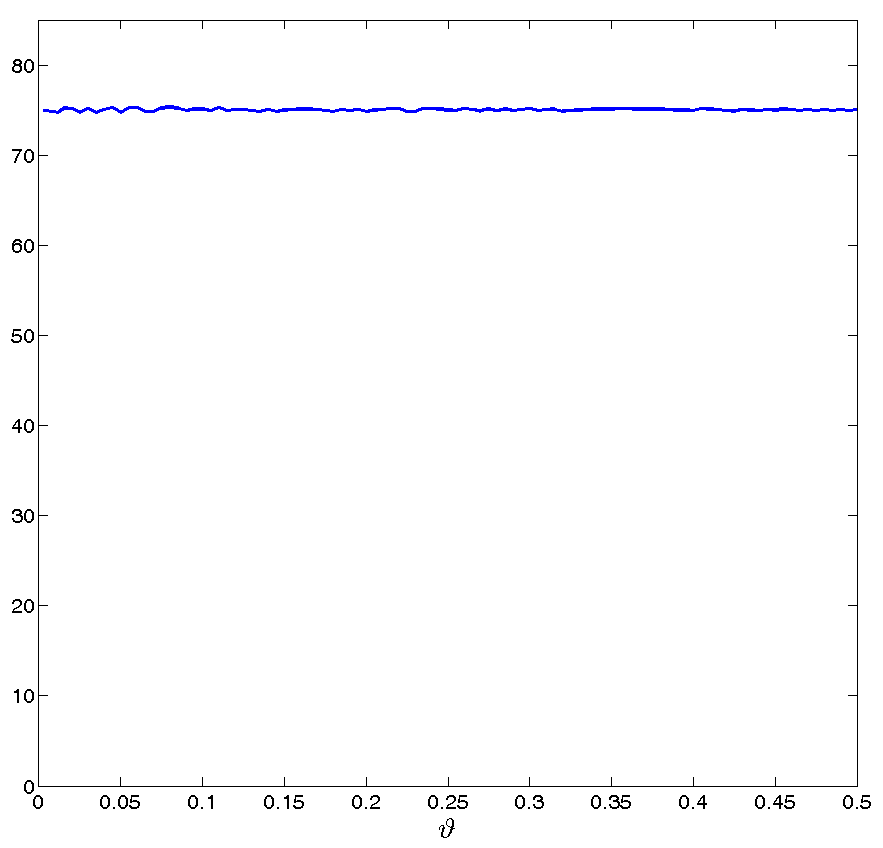}
\caption{The mapping $\theta\mapsto H(\theta,\cx(\theta))$ in Theorem \ref{t:THEOREM_Hauptsatz}(3)}\label{f:FIGURE_AM_PUT_SECOND_COND}
\end{figure}
\begin{figure}
\centering
\includegraphics[width=0.69\textwidth]{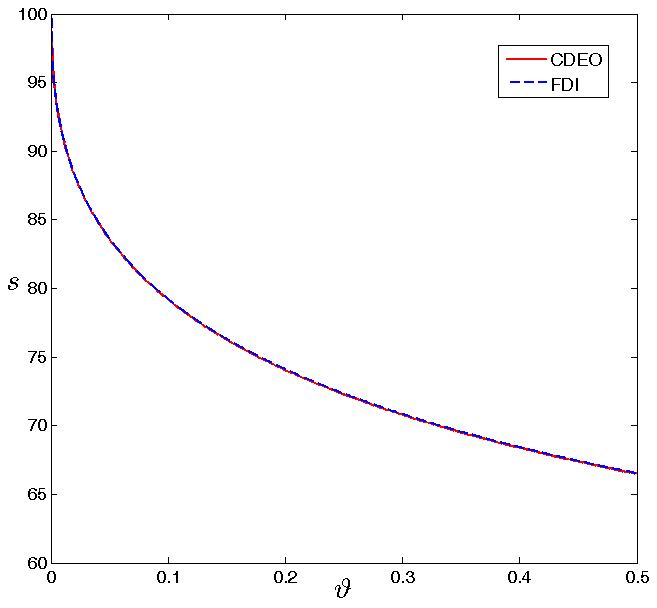}
\caption{Comparison of the CDEO minima curve and a finite difference approximation to the early exercise boundary of the American put}\label{f:FIG_EEC_approximations}
\end{figure}

While Theorem \ref{t:Theorem_summary0} warrants that the American put allows for a cheapest dominating European option in the distributional sense,
Theorem \ref{t:THEOREM_Hauptsatz} does not fully answer the question whether it is actually representable.
Numerically, the CDEO is easily obtained  by semi-infinite linear programming, cf.\ \cite{christensen.14,lenga.17}.
In this section we investigate whether the numerical approximation satisfies the qualitative
assumptions of Theorem \ref{t:THEOREM_Hauptsatz}.

In our numerical experiment we consider the put payoff $g(x) = (e^K - e^x)_+ $ with log-strike price $K=\log100$ and maturity $T=0.5$.
The parameters of the model are chosen as $r=0.06$, $x_0=\log K+0.1$, and $\sigma=0.4$.
Figure \ref{f:FIGURE_CDEO_PUT} displays the price surface of the approximate CDEO along with the put payoff plane.
The $s$-axis represents the stock price $s=e^x$
while the $\theta$-axis indicates the time to maturity of the option. 

If assumption 1 in Theorem \ref{t:THEOREM_Hauptsatz} were violated, we would observe an infinite CDEO price for $\theta>0.5$ and any $s=e^x\in(0,\infty)$.
This is obviously not supported by Figure \ref{f:FIGURE_CDEO_PUT}.
The minima of the functions $x \mapsto v_{\eu, \mu^*}(\theta, x)-g(x)$ for $\theta\in(0,T+\delta)$
are represented by the white curves in Figures \ref{f:FIGURE_CDEO_PUT} and \ref{f:FIGURE_AMPUT_MIN}, using the variable $s=e^x$ instead of $x$.
The graphs are in line with the requirements of assumption 2 in Theorem \ref{t:THEOREM_Hauptsatz}.
The colours in Figure \ref{f:FIGURE_AMPUT_MIN} stand for the level of the function
$(\theta,x) \mapsto v_{\eu, \mu^*}(\theta, x)-g(x)$ that is to be mimimised in $x$ for fixed $\theta$.

The numerical approximation of the function $\theta\to H(\theta, \cx(\theta))$ 
in assumption 3 is shown in Figure \ref{f:FIGURE_AM_PUT_SECOND_COND}.
It stays well away from 0 as required. Given that representability holds, it should in fact have the constant value $2r\sigma^{-2}e^K=75$, which explains
the particular shape in Figure \ref{f:FIGURE_AM_PUT_SECOND_COND}.
If assumption 4 in Theorem \ref{t:THEOREM_Hauptsatz} were violated, we would observe an exploding CDEO price for $\theta\to0$ and $s=e^x=e^K=100$.
The graph rather indicates a vanishing price in the limit -- as is to be expected if the value of the CDEO coincides with the American put price.

Altogether, these qualitative checks indicate that Theorem \ref{t:THEOREM_Hauptsatz} can be applied and hence the put is represented by its CDEO.
This explains not only the close agreement of numerical CDEO and American put values in \cite{christensen.14},
but also the match of the early exercise boundary from a finite difference approximation
and the curve $\theta\mapsto\exp(\cx(\theta))$ suggested by Theorem \ref{t:THEOREM_Hauptsatz}(6), see Figure \ref{f:FIG_EEC_approximations}.

How can these findings be reconciled with the negative result of \cite{jourdain.martini.02}
which states that no sufficiently regular European payoff function can represent the American put?
Using the language of \cite{jourdain.martini.02}, a candidate representing function $\phi$  should satisfy an ordinary differential
equation $\mcA\phi=m$ with an as yet unknown generalised function $m$, where
$\mcA f(x)=(\sigma^2x^2/2)f''(x)+rxf'(x)-rf(x)$.
As stated in \cite[equations (2.1, 2.2)]{jourdain.martini.02}, the general solution to this ODE is of the form
\begin{align}\label{e:jm}
\phi(x)&=ax+bx^{-\alpha}-{2\over\sigma^2}x^{-\alpha}\int_0^xy^\alpha\int_y^Kz^{-2}m(z)dzdy\nonumber\\
&=ax+bx^{-\alpha}-{2\over\sigma^2(\alpha+1)}x^{-\alpha}\int_0^K{(z\wedge x)^{\alpha+1}\over z^{2}}m(z)dz
\end{align}
if 
\begin{equation}\label{e:jmint}
 \int_0^K z^{\alpha-1}|m(z)|dz<\infty,
\end{equation}
where $\alpha:=2r/\sigma^2$, the strike is denoted as $K$, and $a,b\in\rr$ are constants.
For $\phi$ to represent the American put, we need $\phi(x)=0$ for $x\geq K$, which implies $a=0$ and $b>0$.
However, the positivity of $b$ ultimately yields that $\phi$ in \eqref{e:jm} cannot represent the American put, see \cite[Thereom 15]{jourdain.martini.02}.

The integral in \eqref{e:jm} does not make sense if \eqref{e:jmint} is violated. But $\mcA\phi=m$ may still be solved for such $m$, namely by
\begin{align*}\label{e:jm2}
\phi(x)&=ax+bx^{-\alpha}+{2\over\sigma^2}x^{-\alpha}\int_x^Ky^\alpha\int_y^Kz^{-2}m(z)dzdy\\
&=ax+bx^{-\alpha}-{2\over\sigma^2(\alpha+1)}x^{-\alpha}\int_x^K{z^{\alpha+1}-x^{\alpha+1}\over z^{2}}m(z)dz.
\end{align*}
In this case $\phi(x)=0$ for $x\geq K$  implies $a=0=b$, which means that the fateful $b$-term does not appear.
Hence the positive result of our study does not contradict the findings of \cite{jourdain.martini.02} because
the CDEO as a candidate for the representing European claim is not subject 
to the rather strict integrability condition \eqref{e:jmint}.

\section{Conclusion}
As noted in the introduction, the representability of an American option in terms of a
European payoff has several both numerically and conceptually interesting consequences.
In this paper we have made a first step towards verifying that a given American option is representable.
The results of Section \ref{s:SEC_Example_am_put} suggest in particular that
representability holds for the prime example of an American put in the Black-Scholes model,
contrary to the evidence from the analysis in \cite{jourdain.martini.02}.
This gives new hope that the original endeavour of Jourdain and Martini may ultimately
lead to a positive answer and that their concept 
of embedded American options has a broader scope than expected.

As an ambitious goal for future research it remains to fully characterise
representability of American options in the Black-Scholes model and more general
 markets driven by uni- or multivariate diffusions.
In particular, a rigorous proof for the American put is still wanting.

\section*{Acknowledgement}
The authors  thank Josef Teichmann for fruitful discussions and for bringing the
papers \cite{jourdain.martini.01,jourdain.martini.02} to their attention.

\appendix

\section{Auxiliary results}
\begin{lemma}\label{l:stochasticdominance}
Let $X\sim N(\mu_X,\sigma^2)$, $Y\sim N(\mu_Y,\sigma^2)$ be Gaussian random variables with $\mu_X\leq\mu_Y$.
Then the conditional law $P(X\in\cdot|X\geq0)$ is dominated by is counterpart $P(Y\in\cdot|Y\geq0)$ in the usual stochastic order,
i.e.\ $P(X>a|X\geq0)\leq P(Y>a|Y\geq0)$ for any $a\in\rr$ or, equivalently, $E(f(X)|X\geq0)\leq P(f(Y)|Y\geq0)$
for any increasing function $f$ such that the integrals exist.
\end{lemma}
\begin{proof}
It is easy to verify that 
$P(X>a|X\geq0)\leq P(Y>a|Y\geq0)$
holds if and only if $a$ is below some threshold.
This naturally implies the claimed stochastic dominance.
\end{proof}

The following factorisation theorem from multivariate complex analysis gives a sufficient condition for the analytic dependence of zeros.
It is a direct consequence of the Weierstrass preparation theorem. More details can be found in \cite[Chapter 1]{chirka.89}.
\begin{theorem}\label{Theorem_holomorphic_dependence_of_zeros}
For $d\geq 2$ let $f$ be an analytic function on a domain $G = D' \times D \subset \cc^d$ with simply connected $D' \subset \cc^{d-1}$. 
Assume that the function $f(z',\cdot)$ has exactly $m$ distinct zeros in the set $D$ for any $z' \in D'$. 
Then there exist analytic functions
$\alpha_1,...,\alpha_m: D' \to D$, positive integers $k_1,...,k_m$ and an analytic function $\Phi : G \to \cc$ that does not vanish on $G$ such that
\[ f(z',z) = \prod_{l=1}^m \left( z - \alpha_l(z') \right)^{k_l} \Phi(z',z), \quad (z',z)\in G.\]
\end{theorem}
The following version of the analytic implicit function theorem is well suited for our purposes.
It can be obtained as a corollary of Theorem \ref{Theorem_holomorphic_dependence_of_zeros} by applying 
well-known ideas from the proof of Rouch\'e's theorem, cf.\ \cite[page 125]{conway.78}. 
\begin{theorem}\label{THEOREM_Implicit_Function}
For $d\geq 2$ let $f$ be an analytic function on a domain $G = D' \times D  \subset \cc^d$ with simply connected $D' \subset \cc^{d-1}$.
Assume that $f(z'_0,z_0)=0$ and $D_df(z'_0,z_0)\not= 0$ for some $(z'_0,z_0)\in G$.
Then there are open neighbourhoods $U(z_0') \subset D'$ and $V(z_0)\subset D$ of $z_0'$ and $z_0$ 
as well as an analytic function
$g : U(z_0') \to V(z_0) $ such that the equivalence
\[ f(z',z) = 0 \Leftrightarrow z = g(z') \]
holds for all $z' \in U(z_0')$ and $z \in V(z_0)$.
\end{theorem}
\begin{proof} 
Due to $D_df(z'_0,z_0)\not= 0$ there is some $\epsilon_1 > 0$
such that the ball $B(z_0,\epsilon_1)$ is contained in $D$ and $f(z_0',z)\not= 0$ holds for any $z \in B(z_0,\epsilon_1) \setminus \{z_0\}$. 
Moreover, there are constants $c, \epsilon_2 > 0$ such that $B(z_0',\epsilon_2)$ is contained in $D'$ and 
$ \vert f(z',z) \vert > c$ holds for any $z' \in B(z_0',\epsilon_2)$ and any $z \in \cc$ with $\vert z - z_0 \vert = \epsilon_1 $. 
By the choice of $\epsilon_1$ and the argument principle we have
\begin{align*}
\frac{1}{2\pi i}\oint_{\vert z - z_0 \vert = \epsilon_1} \frac{D_df(z_0',z)}{f(z_0',z)} dz = 1.
\end{align*}
The triangle inequality for line integrals yields 
\begin{align*}
\MoveEqLeft{\sup_{\vert z' - z_0' \vert < \epsilon_2 / n} \left\vert  1 - \frac{1}{2\pi i}\oint_{\vert z - z_0 \vert = \epsilon_1} \frac{D_df(z',z)}{f(z',z)} dz   \right\vert}\\
&= \frac{1}{2\pi } \sup_{\vert z' - z_0' \vert < \epsilon_2 / n} \left\vert  \oint_{\vert z - z_0 \vert = \epsilon_1} \frac{D_df(z_0',z)f(z',z)  - D_df(z',z) f(z_0',z)}{f(z_0',z)f(z',z)} dz   \right\vert\\ 
&\leq \alpha \sup_{ \substack {\vert z' - z_0' \vert < \epsilon_2 / n  \\ \vert z - z_0 \vert = \epsilon_1 }}  \left\vert  D_df(z_0',z) f(z',z) -  D_df(z',z) f(z_0',z)  \right\vert ,
\quad n\in\nn
\end{align*}
for  some  $\alpha\in(0,\infty)$ which does not depend on $n$. 
Since $f$ and its derivatives are continuous, we conclude that the right-hand side of this inequality
converges to $0$ as $n$ tends to infinity. Moreover,
$\frac{1}{2\pi i}\oint_{\vert z - z_0 \vert = \epsilon_1} \frac{D_df(z',z)}{f(z',z)} dz $ 
is integer-valued. Consequently there is some $n_0 \in \nn$ such that 
\[ \frac{1}{2\pi i}\oint_{\vert z - z_0 \vert = \epsilon_1} \frac{D_df(z',z)}{f(z',z)} dz    = 1 \]
for any $z' \in B(z_0',\epsilon_2 / n_0)$. Put differently, for any $z' \in B(z_0',\epsilon_2 / n_0)$ the mapping $z \mapsto f(z',z)$ has exactly one 
zero within the set $B(z_0,\epsilon_1)$. By Theorem \ref{Theorem_holomorphic_dependence_of_zeros} there exists an analytic
function $g : B(z_0',\epsilon_2 / n_0) \to B(z_0,\epsilon_1)$ with $ f(z',g(z')) = 0$.
\end{proof}

Proofs of the following results can be found in \cite[Section 5.1]{lenga.17}.
\begin{lemma}\label{LEMMA_norm_and_quotient_of_normal_pdf}
Set 
\begin{equation}\label{e:normpdf}
 \nno(\mu,\sigma^2, y):={1\over\sqrt{2\pi\sigma^2}}\exp\biggl(-{(x-\mu)^2\over2\sigma^2}\biggr).
\end{equation}
\begin{enumerate}
\item 
For any $y\in \rr, \mu, \sigma^2  \in \cc$ with $\Re\sigma^2 > 0$ we have
\begin{align*}
 \vert \nno(\mu,\sigma^2,y) \vert &= \frac{\exp\!\left( 
- \frac { \Re \sigma^2}{2\vert\sigma^2\vert^2}  \left( y -   \Re\mu - \frac{ \Im\mu \Im\sigma^2 }{\Re\sigma^2}  \right)^2 
  +\frac{(\Im \mu)^2}{2 \Re \sigma^2 }  \right) }{ \sqrt{2\pi\vert \sigma^2 \vert}}.
\end{align*}
\item For any $\mu, \widetilde{\mu} \in \cc$ and $\sigma, \widetilde{\sigma} \in \cc \setminus \{0\}$ with $\sigma \not= \widetilde{\sigma}$ we have
\[ \frac{\nno(\mu,\sigma^2,y)}{\nno(\widetilde{\mu}, \widetilde{\sigma}^2,y)} = \frac{ \widetilde{\sigma} }{\sigma}\exp\!\left( -\frac{(y - A)^2}{2B} \right) 
\exp\!\left( - \frac{(\mu - \widetilde{\mu} )^2}{2(\sigma^2 - \widetilde{\sigma}^2)}  \right)\]  
with
\begin{align}\label{e:AB}
&A:= \frac{\widetilde{\mu} \sigma^2 - \mu \widetilde{\sigma}^2}{\sigma^2 - \widetilde{\sigma}^2},
&B:= \frac{\widetilde{\sigma}^2 \sigma^2}{\widetilde{\sigma}^2 - \sigma^2}.
\end{align}
\item For any $\mu, \widetilde{\mu} \in \cc$ and $\sigma \in \cc \setminus \{0\}$ we have
\[ \frac{\nno(\mu,\sigma^2,y)}{\nno(\widetilde{\mu}, \sigma^2,y)} = \exp\!\left( y \frac{ \mu - \widetilde{\mu} }{\sigma^2 } \right) \exp\!\left( \frac{\widetilde{\mu}^2 - \mu^2}{2\sigma^2 }  \right).\]
\end{enumerate}
\end{lemma}

\begin{lemma}\label{LEMMA_analyticity_sol_heat_eaqtion}
For $\widehat{r}\in\rr, \sigma>0$ and any measure $\mu \in \M^+(\rr)$ we define the generalised European value function
\[V(\theta,x):= \int \nno\bigl(x +\widehat{r} \theta, \sigma^2 \theta, y \bigr)  \mu(dy),\]
where $\phi$ is defined as in \eqref{e:normpdf}.
Suppose there exists some $(T,x_0)\in (0,\infty)\times\rr$ with $V(T,x_0)<\infty$.
Then the mapping $V$ is analytic on the open $\cc^2$-domain
\[\left\{ \theta \in \cc :  \sqrt{ (\Re \theta - T/2)^2 + (\Im \theta)^2 } < T/2 \right\} \times \cc.\]
\end{lemma}

\begin{lemma}\label{l:kolmogorov}
Let $\mu\in\M^+(\rr)$, $T>0$, $\sigma^2>0$, $\widehat r=r-\sigma^2/2$, $\mathscr A=(\sigma^2/2) D_{22}+\hat r D_2-r$.
Then 
\[\Psi(\theta,x):=\int{\nno(x+\hat r\theta,\sigma^2\theta,y)\over\nno(x+\hat rT,\sigma^2T,y)}\mu(dy)\]
satisfies $(D_1-\mathscr A)\Psi=0$ on $(0,T)\times\rr$.
\end{lemma}
\begin{proof}
$\psi(\theta,x)=\nno(x+\hat r\theta,\sigma^2\theta,y)$ satisfies $(D_1-\mathscr A)\psi=0$ for fixed $y$.
The claim follows from interchanging differentiation and integration.
\end{proof}

\bibliographystyle{plain}
\bibliography{literatur}

\begin{thebibliography}{10}

\bibitem{aliprantis.border.06}
C.~Aliprantis and K.~Border.
\newblock {\em Infinite dimensional analysis}.
\newblock Springer, Berlin, third edition, 2006.

\bibitem{chirka.89}
E.~Chirka.
\newblock {\em Complex analytic sets}.
\newblock Kluwer, Dordrecht, 1989.

\bibitem{christensen.14}
S.~Christensen.
\newblock A method for pricing {A}merican options using semi-infinite linear
  programming.
\newblock {\em Math. Finance}, 24(1):156--172, 2014.

\bibitem{conway.78}
J.~Conway.
\newblock {\em Functions of one complex variable}.
\newblock Springer, New York, second edition, 1978.

\bibitem{davis.karatzas.94}
M.~Davis and I.~Karatzas.
\newblock A deterministic approach to optimal stopping.
\newblock In {\em Probability, statistics and optimisation}. Wiley, Chichester,
  1994.

\bibitem{haugh.kogan.04}
M.~Haugh and L.~Kogan.
\newblock Pricing {A}merican options: a duality approach.
\newblock {\em Oper. Res.}, 52(2):258--270, 2004.

\bibitem{jourdain.martini.01}
B.~Jourdain and C.~Martini.
\newblock American prices embedded in {E}uropean prices.
\newblock {\em Ann. Inst. H. Poincar\'{e} Anal. Non Lin\'{e}aire}, 18(1):1--17,
  2001.

\bibitem{jourdain.martini.02}
B.~Jourdain and C.~Martini.
\newblock Approximation of {A}merican put prices by {E}uropean prices via an
  embedding method.
\newblock {\em Ann. Appl. Probab.}, 12(1):196--223, 2002.

\bibitem{klenke.14}
A.~Klenke.
\newblock {\em Probability theory}.
\newblock Springer, London, second edition, 2014.

\bibitem{krantz.92}
S.~Krantz.
\newblock {\em Function theory of several complex variables}.
\newblock AMS Chelsea Publishing, Providence, RI, 1992.

\bibitem{lamberton.09}
D.~Lamberton.
\newblock Optimal stopping and {A}merican options.
\newblock Lecture notes, 2009.

\bibitem{lang.93}
S.~Lang.
\newblock {\em Real and functional analysis}.
\newblock Springer, New York, third edition, 1993.

\bibitem{lenga.17}
M.~Lenga.
\newblock {\em Representable options}.
\newblock PhD thesis, Kiel University, 2017.
\newblock \url{https://nbn-resolving.org/urn:nbn:de:gbv:8-diss-210007}.

\bibitem{lerche.urusov.07}
R.~Lerche and M.~Urusov.
\newblock Optimal stopping via measure transformation: the {B}eibel-{L}erche
  approach.
\newblock {\em Stochastics}, 79(3-4):275--291, 2007.

\bibitem{meise.vogt.97}
R.~Meise and D.~Vogt.
\newblock {\em Introduction to functional analysis}.
\newblock The Clarendon Press, Oxford University Press, New York, 1997.

\bibitem{nachman.88}
D.~Nachman.
\newblock Spanning and completeness with options.
\newblock {\em Rev. Financ. Stud.}, 1(3):311--328, 1988.

\bibitem{peskir.shiryaev.06}
G.~Peskir and A.~Shiryaev.
\newblock {\em Optimal stopping and free-boundary problems}.
\newblock Birkh\"{a}user, Basel, 2006.

\bibitem{rockafellar.74}
T.~Rockafellar.
\newblock {\em Conjugate duality and optimization}.
\newblock Society for Industrial and Applied Mathematics, Philadelphia, Pa.,
  1974.

\bibitem{rogers.02}
C.~Rogers.
\newblock Monte {C}arlo valuation of {A}merican options.
\newblock {\em Math. Finance}, 12(3):271--286, 2002.

\bibitem{rudin.87}
W.~Rudin.
\newblock {\em Real and complex analysis}.
\newblock McGraw-Hill., New York, third edition, 1987.

\end{thebibliography}

\end{document}